\documentclass[final,leqno]{siamltex704}

\usepackage{amsmath,amssymb}
\usepackage{exscale,cmmib57}
\usepackage{amsfonts}

\usepackage[pdftex]{graphicx}
\usepackage{url}

\newtheorem{remark}[theorem]{Remark}

\usepackage{mathrsfs}
\renewcommand{\mathcal}[1]{\mathscr{#1}}

\newcommand{\AC}{\mathsf{AC}}

\DeclareMathOperator{\dt}{dt}

\newcommand{\PP}{\mathsf{PP}}

\DeclareMathOperator{\degthr}{deg_{\pm}}
\newcommand{\OMB}{\text{\rm OMB}}

\usepackage{theory}

\newcommand{\QED}{\vbox{\hrule height0.6pt\hbox{%
   \vrule height1.3ex width0.6pt\hskip0.8ex
   \vrule width0.6pt}\hrule height0.6pt
  }}

\newenvironment{annotatedproof}[1]
{\par{\it #1}. \ignorespaces}
{}

\title{The Pattern Matrix Method%
  \thanks{To appear in \emph{SIAM~J.~Comput.}, 2009.
  A preliminary version of this article appeared under the title
  ``The Pattern Matrix Method for Lower Bounds on Quantum Communication''
  in \emph{Proceedings of the {\rm 40}th Annual ACM Symposium on Theory of
  Computing} (STOC), pages 85-94, 2008.}
}

\author{
Alexander~A.~Sherstov\thanks{Department of Computer Science,
The University of Texas at Austin
({\tt sherstov@cs.utexas.edu}).}
}

\renewcommand{\AA}{\mathcal{A}}
\newcommand{\BB}{\mathcal{B}}

\newcommand{\CC}{\mathcal{C}}

\newcommand{\VV}{\mathcal{V}}

\newcommand{\1}{\mathbf{1}}
\newcommand{\tr}{^{\sf T}}
\newcommand{\F}{_\mathrm{F}}

\DeclareMathOperator{\Span}{span}
\DeclareMathOperator{\adeg}{\deg_{1/3}}
\DeclareMathOperator{\degeps}{\deg_{\epsilon}}
\newcommand{\qcc}{Q_{1/3}^*}   
 
\newcommand{\final}{\text{\sf Final}} 
\newcommand{\init}{\text{\sf Initial}}

\newcommand{\oneeighth}{{\textstyle\frac{1}{8}}}

\newcommand{\mip}{\mathrm{MP}}

\begin{document}
\maketitle

\begin{abstract} 
We develop a novel and powerful technique for communication lower
bounds, the \emph{pattern matrix method}.  Specifically, fix an
arbitrary function $f\colon \zoon\to\zoo$ and let $A_f$ be the
matrix whose columns are each an application of $f$ to some subset
of the variables $x_1,x_2,\dots,x_{4n}.$ We prove that $A_f$ has
bounded-error communication complexity $\Omega(d),$ where $d$ is
the approximate degree of $f.$ This result remains valid in the
quantum model, regardless of prior entanglement. In particular, it
gives a new and simple proof of Razborov's breakthrough quantum lower
bounds for disjointness and other symmetric predicates.  We further
characterize the discrepancy, approximate rank, and approximate
trace norm of $A_f$ in terms of well-studied analytic properties
of $f,$ broadly generalizing several recent results on small-bias
communication and agnostic learning.  The method of this paper has
recently enabled important progress in multiparty communication
complexity.
\end{abstract}



\pagestyle{myheadings}
\thispagestyle{empty}
\markboth{ALEXANDER A. SHERSTOV}{THE PATTERN MATRIX METHOD}

\section{Introduction} \label{sec:intro}

A central model in communication complexity is the \emph{bounded-error
model}. Let $f\colon X\times Y\to\moo$ be a given function, where
$X$ and $Y$ are finite sets.  Alice receives an input $x\in X,$
Bob receives $y\in Y,$ and their objective is to compute $f(x,y)$
with minimal communication.  To this end, Alice and Bob share an
unlimited supply of random bits. Their protocol is said to \emph{compute
$f$} if on every input $(x,y),$ the output is correct with probability
at least $1-\epsilon.$ The canonical setting is $\epsilon=1/3,$ but
any other parameter $\epsilon\in(0,1/2)$ can be considered.
The \emph{cost} of a protocol is the worst-case number
of bits exchanged on any input.  Depending on the physical nature
of the communication channel, one studies the \emph{classical model},
in which the messages are classical bits $0$ and $1,$ and the more
powerful \emph{quantum model}, in which the messages are quantum
bits and arbitrary prior entanglement is allowed. The communication
complexity in these models is denoted $R_{\epsilon}(f)$ and
$Q^*_{\epsilon}(f),$ respectively.

Bounded-error protocols have been the focus of much research in communication
complexity since the introduction of the area by Yao~\cite{yao79cc}
three decades ago.  A variety of techniques have been developed
for proving lower bounds on classical communication, 
e.g.,~\cite{KS92disj,
razborov90disj,
GHR92,
raz95fourier-cc,
chakrabarti-et-al01info-complexity,
linial07factorization-stoc,
gkkrw07quantum-classical,
sherstov07ac-majmaj}.
There has been consistent progress on quantum communication as 
well~\cite{yao93quantum,
ambainis03sampling,  
buhrman-dewolf01polynomials, 
klauck01interaction, 
klauck01quantum,     
razborov03quantum, 
linial07factorization-stoc}, 
although quantum protocols remain much less understood than 
their classical counterparts. 

The main contribution of this paper is a novel and powerful method
for lower bounds on classical and quantum communication complexity,
the \emph{pattern matrix method}. The method converts 
analytic properties of Boolean functions into lower bounds for the
corresponding communication problems. The analytic properties in
question pertain to the approximation and sign-representation of a
given Boolean function by real polynomials of low degree, which are
among the oldest and most studied objects in theoretical computer
science. In other words, the pattern matrix method takes the wealth
of results available on the representations of Boolean functions by
real polynomials and puts them at the disposal of communication
complexity.

We consider two ways of representing Boolean functions by real
polynomials.  Let $f\colon\zoon\to\moo$ be a given Boolean function.
The \emph{$\epsilon$-approximate degree} of $f,$ denoted $\degeps(f),$
is the least degree of a real polynomial $p$ such that
\mbox{$|f(x)-p(x)|\leq \epsilon$} for all $x\in\zoon.$ There is an
extensive literature on the $\epsilon$-approximate degree of Boolean
functions%
~\cite{nisan-szegedy94degree,
paturi92approx,
kahn96incl-excl,
buhrman-et-al99small-error,
aaronson-shi04distinctness,
ambainis05collision,
sherstov07inclexcl-ccc,
de-wolf08approx-degree},
for the canonical setting $\epsilon=1/3$ and various other settings.
Apart from uniform approximation, the other representation scheme
of interest to us is sign-representation. Specifically, the
\emph{degree-$d$ threshold weight} $W(f,d)$ of $f$ is the minimum
$\sum_{|S|\leq d} |\lambda_S|$ over all integers $\lambda_S$ such
that
\begin{align*}
f(x)\equiv \sign\left( \sum_{S\subseteq\{1,\dots,n\}, \, |S|\leq d}
\lambda_S\chi_S(x)\right),
\end{align*}
where $\chi_S(x) = (-1)^{\sum_{i\in S}x_i}.$ If no such integers
$\lambda_S$ exist, we write $W(f,d)=\infty.$ The threshold weight
of Boolean functions has been heavily studied, both when $W(f,d)$
is infinite%
~\cite{minsky88perceptrons,
aspnes91voting,
krause94depth2mod,
KP98threshold,
KS01dnf,
KOS:02,
odonnell03degree}
and when it is finite%
~\cite{minsky88perceptrons,
myhill-kautz61,
beigel94perceptrons,
vereshchagin95weight,
klivans-servedio06decision-lists,
mlj07sq,
podolskii07perceptrons,
podolskii08perceptrons}.
The notions of uniform approximation and sign-representation are
closely related, as we discuss in Section~\ref{sec:prelim}. Roughly
speaking, the study of threshold weight corresponds to the study
of the $\epsilon$-approximate degree for $\epsilon=1-o(1).$

Having defined uniform approximation and sign-representation for
Boolean functions, we now describe how we use them to prove
communication lower bounds. The central concept in our work is what
we call a \emph{pattern matrix}. Consider the
communication problem
of computing
\[ f(x|_V), \]
where $f\colon\zoo^t\to\moo$ is a fixed Boolean function; the string
$x\in\zoon$ is Alice's input ($n$ is a multiple of $t$); and the
set $V\subset\{1,2,\dots,n\}$ with $|V|=t$ is Bob's input.  In
words, this communication problem corresponds to a situation when
the function $f$ depends on only $t$ of the inputs $x_1,\dots,x_n.$
Alice knows the values of all the inputs $x_1,\dots,x_{n}$ but does
not know which $t$ of them are relevant.  Bob, on the other hand,
knows which $t$ inputs are relevant but does not know their values.
This communication game was introduced and studied in an earlier
work by the author~\cite{sherstov07ac-majmaj}, in the context of
small-bias communication. For the purposes of the introduction,
one can think of the \emph{$(n,t,f)$-pattern matrix} as
the matrix $[f(x|_V)]_{x,V},$ where $V$ ranges over the $(n/t)^t$
sets that have exactly one element from each block of the following
partition:
\begin{align*}
  \{1,\dots,n\} = \left\{1,  2,  \dots,   \frac{n}{t}
  \rule{0mm}{5mm}
  \right\} 
     \cup 
  \left\{\frac{n}{t}+1,   \dots,  \frac{2n}{t}
  \right\}
     \cup 
	 \cdots 
	 \cup
  \left\{\frac{(t-1)n}{t}+1,  \dots,  n\right\}.
\end{align*}
We defer the precise definition to Section~\ref{sec:pattern-matrices}.
Observe that restricting $V$ to be of special form 
only makes our results stronger.

\subsection{Our results}

Our main result is a lower bound on the communication complexity of a pattern
matrix in terms of the $\epsilon$-approximate degree of the base function
$f.$ The lower bound holds for both classical and quantum protocols,
regardless of prior entanglement.

\begin{theorem}[communication complexity]
Let $F$ be the $(n,t,f)$-pattern matrix, 
where $f\colon \zoo^t\to\moo$ is given.
Then for every $\epsilon\in[0,1)$ and every $\delta<\epsilon/2,$
\begin{align}
Q^*_{\delta}(F) &\geq
\frac{1}{4}  \degeps(f)\log_2 \left(\frac{n}{t}\right) - 
\frac12 \log_2\left(\frac{3}{\epsilon-2\delta}\right).
	\nonumber
\intertext{In particular,}
Q^*_{1/7}(F) &>
\frac{1}{4}  \deg_{1/3}(f)\log_2 \left(\frac{n}{t}\right) - 3. 
	\label{eqn:main-cc-bounded}
\end{align}
\label{thm:main-cc}
\end{theorem}

Note that Theorem~\ref{thm:main-cc} yields lower bounds for
communication complexity with error probability $\delta$ for any
$\delta\in(0,1/2).$ In particular, apart from bounded-error
communication~(\ref{eqn:main-cc-bounded}), we obtain lower bounds
for communication with small bias, i.e., error probability
$\frac12-o(1).$ In Section~\ref{sec:additional-results}, we derive
another lower bound for small-bias communication, in terms of
threshold weight $W(f,d).$

As R.~de Wolf pointed out to us~\cite{dewolf-personal-oct-2007}, 
\label{rem:dewolf}
the lower bound (\ref{eqn:main-cc-bounded}) for bounded-error
communication is within a polynomial of optimal.  More precisely,
$F$ has a classical deterministic protocol
with cost $O(\deg_{1/3}(f)^6\log (n/t)),$ by the results of Beals
et al.~\cite{beals-et-al01quantum-by-polynomials}.  See
Proposition~\ref{prop:det-upper-bound} for details.  In particular,
Theorem~\ref{thm:main-cc} exhibits a large new class of communication
problems $F$ whose quantum communication complexity is 
polynomially related to their classical complexity, \emph{even}
if prior entanglement is allowed. Before our work, the largest class
of problems with polynomially related quantum and classical
bounded-error complexities was the class of symmetric functions
(see Theorem~\ref{thm:razborov03quantum} below), which is broadly
subsumed by Theorem~\ref{thm:main-cc}.  Exhibiting a polynomial
relationship between the quantum and classical bounded-error
complexities for \emph{all} functions $F\colon X\times Y\to\moo$
is a longstanding open problem.

Pattern matrices are of interest because they occur as submatrices in many
natural communication problems. For example, Theorem~\ref{thm:main-cc}
can be interpreted in terms of function composition. 
Setting $n=4t$ for concreteness, we obtain:
\begin{corollary}
Let $f\colon \zoo^t\to\moo$ be given. 
Define $F\colon \zoo^{4t}\times\zoo^{4t}\to\moo$ by
$
F(x,y) = f(\dots,
(x_{i,1}y_{i,1}\,\vee\,
x_{i,2}y_{i,2}\,\vee\,
x_{i,3}y_{i,3}\,\vee\,
x_{i,4}y_{i,4})
,\dots).
$
Then
\begin{align*}
Q^*_{1/7}(F) > \frac{1}{4}\adeg(f)- 3.
\end{align*}
\label{cor:main}
\end{corollary}

As another illustration of Theorem~\ref{thm:main-cc}, we revisit the
quantum communication complexity of symmetric functions. In this
setting Alice has a string $x\in\zoon,$ Bob has a string $y\in\zoon,$
and their objective is to compute $D(\sum x_iy_i)$ for some predicate
$D\colon \zodn\to\moo$ fixed in advance.  This framework encompasses
several familiar functions, such as {\sc disjointness} (determining
if $x$ and $y$ intersect) and {\sc inner product modulo $2$}
(determining if $x$ and $y$ intersect in an odd number of positions).
In a celebrated result, Razborov~\cite{razborov03quantum} established
optimal lower bounds on the quantum communication complexity of
every function of the above form:

\begin{theorem}[Razborov]
Let $D\colon \zodn\to\moo$ be a given predicate.
Put $f(x,y)=D(\sum x_iy_i).$ Then
\begin{align*}
Q^*_{1/3}(f) \geq
\Omega(\sqrt{n\ell_0(D)}+\ell_1(D)),
\end{align*}
where $\ell_0(D)\in\{0,1,\dots,\lfloor n/2\rfloor\}$ and
$\ell_1(D)\in\{0,1,\dots,\lceil n/2\rceil\}$
are the smallest integers such that $D$ is constant in the range
$[\ell_0(D),n-\ell_1(D)].$ 
\label{thm:razborov03quantum}
\end{theorem}

Using Theorem~\ref{thm:main-cc}, we give a new and simple proof of
Razborov's result. No alternate proof was available prior to this
work, despite the fact that this problem has drawn the attention
of various researchers%
~\cite{ambainis03sampling,
buhrman-dewolf01polynomials, 
klauck01interaction, 
klauck01quantum, 
hoyer-dewolf02disjointness,
linial07factorization-stoc}.  
Moreover, the next-best lower bounds for general predicates were
nowhere close to Theorem~\ref{thm:razborov03quantum}.  To illustrate,
consider the disjointness predicate  $D,$ given by $D(t)=1\Leftrightarrow
t=0.$ Theorem~\ref{thm:razborov03quantum} shows that it has
communication complexity $\Omega(\sqrt n),$ while the next-best
lower bound~\cite{ambainis03sampling, buhrman-dewolf01polynomials}
was only $\Omega(\log n).$

\paragraph{Approximate rank and trace norm}
We now describe some matrix-analytic consequences of our work.
The $\epsilon$-approximate rank of a matrix $F\in\moo^{m\times n},$
denoted $\rk_\epsilon F,$ is the least rank of a real matrix $A$
such that $|F_{ij}-A_{ij}|\leq\epsilon$ for all $i,j.$ This natural
analytic quantity arose in the study of quantum
communication~\cite{yao93quantum, buhrman-dewolf01polynomials,
razborov03quantum} and has since found applications to learning
theory.  In particular, Klivans and Sherstov~\cite{colt07rankeps}
proved that concept classes (i.e., sign matrices) with high approximate
rank are beyond the scope of all known techniques for efficient
learning, in Kearns' well-studied \emph{agnostic
model}~\cite{kearns-et-al94agnostic}.  Exponential lower bounds
were derived in~\cite{colt07rankeps} on the approximate rank of
disjunctions, majority functions, and decision lists, with the corresponding
implications for agnostic learning.  We broadly generalize these
results on approximate rank to \emph{any} functions with high
approximate degree or high threshold weight:

\begin{theorem}[approximate rank]
Let $F$ be the $(n,t,f)$-pattern matrix, 
where $f\colon \zoo^t\to\moo$ is given.
Then for every $\epsilon\in[0,1)$ and every $\delta\in[0,\epsilon],$
\begin{align*}
\rk_\delta F &\geq 
\left(\frac{\epsilon - \delta}{1 + \delta}\right)^2 
\left(\frac nt\right)^{\degeps(f)}.
\label{eqn:main-rank}
\intertext{In addition, for every $\gamma\in(0,1)$ and every integer
$d\geq1,$} 
\rk_{1-\gamma} F &\geq 
\left(\frac{\gamma}{2-\gamma}\right)^2
\min\left\{\left(\frac nt\right)^d,
\frac{W(f,d-1)}{2t}
         \right\}.
\end{align*}
  \label{thm:main-approx-rank}
\end{theorem}

We derive analogous results for the \emph{approximate trace norm,}
another matrix-analytic notion that has been studied in complexity
theory.  Theorem~\ref{thm:main-approx-rank} is close to optimal for
a broad range of parameters.  See Section~\ref{sec:approx-rank}
for details.

\paragraph{Discrepancy}
The discrepancy of a function $F\colon X\times Y\to\moo,$ denoted
$\disc(F),$ is a combinatorial measure of the complexity of $F$
(small discrepancy corresponds to high complexity).
This complexity measure plays a central role in the study of
communication.  In particular, it fully characterizes membership
in $\PP^{cc},$ the class of communication problems with efficient
small-bias protocols~\cite{klauck01quantum-journal}. Discrepancy
is also known~\cite{LS08learning-cc} be to equivalent to \emph{margin
complexity}, a key notion in learning theory. Finally, discrepancy
is of interest in circuit complexity~\cite{GHR92,
hajnal93threshold-const-depth, nisan93threshold}.  We are able to
characterize the discrepancy of every pattern matrix in terms of
threshold weight:

\begin{theorem}[discrepancy]
Let $F$ be the $(n,t,f)$-pattern matrix, 
for a given function $f\colon \zoo^t\to\moo.$ Then
\begin{align*}
\disc(F) \leq \min_{d=1,\dots,t}
\max\left\{ \left(\frac {2t}{W(f,d-1)}\right)^{1/2}, 
                             \left(\frac tn\right)^{d/2}
         \right\}.
\end{align*}
\label{thm:main-discrepancy}
\end{theorem}

As we show in Section~\ref{sec:disc}, Theorem~\ref{thm:main-discrepancy}
is close to optimal. It is a substantial improvement on the author's
earlier work~\cite{sherstov07ac-majmaj}.

As an application of Theorem~\ref{thm:main-discrepancy}, we revisit
the discrepancy of $\AC^0,$ the class of polynomial-size constant-depth
circuits with AND, OR, NOT gates. In an earlier
work~\cite{sherstov07ac-majmaj}, we obtained the first exponentially
small upper bound on the discrepancy of a function in $\AC^0.$ We
used this result in~\cite{sherstov07ac-majmaj} to prove that depth-$2$
majority circuits for $\AC^0$ require exponential size, solving an
open problem due to Krause and Pudl{\'a}k~\cite{krause94depth2mod}.
Using Theorem~\ref{thm:main-discrepancy}, we are able to considerably
sharpen the bound in~\cite{sherstov07ac-majmaj}. Specifically, we
prove:

\begin{theorem}
Let $f(x,y)=
	\bigvee_{i=1}^m
	\bigwedge_{j=1}^{m^2}
	  (x_{ij}\vee y_{ij}).$ 
Then
\begin{align*}
\disc(f) = \exp\{-\Omega(m)\}. 
\end{align*}
   \label{thm:main-mip}
\end{theorem}

We defer the new circuit implications and other discussion to 
Sections~\ref{sec:disc} and~\ref{sec:app-discrepancy}.
Independently of the work in~\cite{sherstov07ac-majmaj}, Buhrman
et al.~\cite{buhrman07pp-upp} exhibited another function in $\AC^0$
with exponentially small discrepancy:

\par{\sc Theorem} (Buhrman et al.). {\it
Let $f:\zoon\times\zoon\to\moo$ be given by 
$f(x,y)= \sign\left(1 + \sum_{i=1}^n (-2)^i x_iy_i\right).$
Then 
\begin{align*}
\disc(f) = \exp\{-\Omega(n^{1/3})\}.
\end{align*}
}

Using Theorem~\ref{thm:main-discrepancy}, we give a new and simple
proof of this result.

\subsection{Our techniques}

The setting in which to view our work is the \emph{generalized
discrepancy method,} a straightforward but very useful principle
introduced by Klauck~\cite{klauck01quantum} and reformulated in its
current form by Razborov~\cite{razborov03quantum}. Let $F(x,y)$ be
a Boolean function whose bounded-error communication complexity is
of interest.  The generalized discrepancy method asks for a Boolean
function $H(x,y)$ and a distribution $\mu$ on $(x,y)$-pairs such
that:
\begin{itemize}
\item[(1)] the functions $F$ and $H$ have correlation $\Omega(1)$ 
under $\mu$;  and 
\item[(2)] all low-cost protocols
have negligible advantage in computing $H$ under $\mu.$
\end{itemize}
If such $H$ and $\mu$ indeed exist, it follows that no low-cost
protocol can compute $F$ to high accuracy (otherwise it would be a
good predictor for the hard function $H$ as well). This method
applies broadly to many models of communication, as we discuss in
Section~\ref{sec:discrepancy}.  It generalizes Yao's original
discrepancy method~\cite{ccbook}, in which $H=F.$ The advantage of
the generalized version is that it makes it possible, in theory,
to prove lower bounds for functions such as {\sc disjointness,} to
which the traditional method does not apply.

The hard part, of course, is finding $H$ and $\mu$ with the desired
properties. Except in rather restricted
cases~\cite[Thm.~4]{klauck01quantum}, it was not known how to do
it. As a result, the generalized discrepancy method was of limited
practical use prior to this paper. Here we overcome this difficulty,
obtaining $H$ and $\mu$ for a broad range of problems, namely, the
communication problems of computing $f(x|_V).$

Pattern matrices are a crucial first ingredient of our solution.
We derive an exact, closed-form expression for the singular values
of a pattern matrix and their multiplicities. This spectral information
reduces our search from $H$ and $\mu$ to a much smaller and simpler
object, namely, a function $\psi\colon \zoo^t\to\Re$ with certain
properties.  On the one hand, $\psi$ must be well-correlated with
the base function $f.$ On the other hand, $\psi$ must be orthogonal
to all low-degree polynomials.  We establish the existence of such
$\psi$ by passing to the \emph{linear programming dual} of the
approximate degree of~$f.$ Although the approximate degree and
its dual are classical notions, we are not aware of any previous
use of this duality to prove communication lower bounds.

For the results that feature threshold weight, we combine the above
program with the dual characterization of threshold weight.  To
derive the remaining results on approximate rank, approximate trace
norm, and discrepancy, we apply our main technique along with several
additional matrix-analytic and combinatorial arguments.


\subsection{Recent work on multiparty complexity} 
The method of this paper has recently enabled important progress
in multiparty communication complexity by a number of researchers.
Lee and Shraibman~\cite{lee-shraibman08disjointness} and Chattopadhyay
and Ada~\cite{chatt-ada08disjointness} observed that our method
adapts in a straightforward way to the multiparty model, thereby
obtaining much improved lower bounds on the communication complexity of
{\sc disjointness} for up to $\log\log n$ players.  David and
Pitassi~\cite{pitassi08np-rp} ingeniously combined this line of
work with the probabilistic method, establishing a separation of
the communication classes {\sf NP$^{cc}_k$} and {\sf BPP$^{cc}_k$}
for up to $k=(1-\epsilon)\log n$ players.  Their construction was
derandomized in a follow-up paper by David, Pitassi, and
Viola~\cite{david-pitassi-viola08bpp-np}, resulting in an explicit
separation.  See the survey article~\cite{dual-survey} for a unified
guide to these results, complete with all the key proofs.  A
very recent development is due to Beame and
Huynh-Ngoc~\cite{beame-huyn-ngoc08multiparty-eccc}, who continue
this line of research with improved multiparty lower bounds for
$\AC^0$ functions.

\subsection{Organization}
We start with a thorough review of technical preliminaries in
Section~\ref{sec:prelim}. The two sections that follow are concerned
with the two principal ingredients of our technique, the pattern
matrices and the dual characterization of the approximate degree
and threshold weight.  Section~\ref{sec:pattern-matrix-method}
integrates them into the generalized discrepancy method and establishes
our main result, Theorem~\ref{thm:main-cc}. In
Section~\ref{sec:additional-results}, we prove an additional version
of our main result using threshold weight.  We characterize the
discrepancy of pattern matrices in Section~\ref{sec:disc}.
Approximate rank and approximate trace norm are studied next, in
Section~\ref{sec:approx-rank}. We illustrate our main result in
Section~\ref{sec:razborovs-result} by giving a new proof
of Razborov's quantum lower bounds. As another illustration, we
study the discrepancy of $\AC^0$ in Section~\ref{sec:app-discrepancy}.
We conclude with some remarks on the well-known log-rank conjecture
in Section~\ref{sec:logrank} and a discussion of related work in
Section~\ref{sec:shi-zhu}.

\section{Preliminaries} \label{sec:prelim}

%
%

We view Boolean functions as mappings $X\to\moo$ for a finite set $X,$
where $-1$ and $1$ correspond to ``true'' and ``false,'' respectively.
Typically, the domain will be $X=\zoon$ or $X=\zoon\times\zoon.$ A
\emph{predicate} is a mapping $D\colon \zodn\to\moo.$ The
notation $[n]$ stands for the set $\{1,2,\dots,n\}.$ For a set
$S\subseteq[n],$ its \emph{characteristic vector} $\1_S\in\zoon$ is
defined by
\[ (\1_S)_i=
\begin{cases}
1 & \text{if $i\in S,$} \\
0 & \text{otherwise.}
\end{cases}
\]
For $b\in\zoo,$ we put $\neg b=1-b.$ For $x\in\zoon,$ we define
$|x| = x_1+\cdots+x_n.$ For $x,y\in\zoon,$ the notation $x\wedge
y\in\zoon$ refers as usual to the component-wise conjunction of $x$
and $y.$ Analogously, the string $x\vee y$ stands for 
the component-wise disjunction of $x$ and $y.$
In particular, $|x\wedge y|$ is the number of
positions in which the strings $x$ and $y$ both have a $1.$ Throughout
this manuscript, ``$\log$" refers to the logarithm to base $2.$ As
usual, we denote the base of the natural logarithm by $\e=2.718\dots.$
For any mapping $\phi\colon X\to\Re,$ where $X$ is a finite set,
we adopt the standard notation $\|\phi\|_\infty = \max_{x\in
X}|\phi(x)|.$ We adopt the standard definition of the sign function:
\begin{align*}
\sign t = 
\begin{cases}
   -1 & \text{if $t<0,$}\\
   0 & \text{if $t=0,$}\\
   1 & \text{if $t>0.$}
\end{cases}
\end{align*}

Finally, we recall the Fourier transform over $\Z_2^n.$ Consider the
vector space of functions $\zoon\to\Re,$ equipped with the inner
product
\[\langle f,g\rangle = 2^{-n} \sum_{x\in\zoon}f(x)g(x).\]
For $S\subseteq[n],$ define $\chi_S\colon \zoon\to\moo$ by
$\chi_S(x) =(-1)^{\sum_{i\in S} x_i}.$
Then $\{\chi_S\}_{S\subseteq[n]}$ is an orthonormal basis for the
inner product space in question.  As a result, every function
$f\colon \zoon\to\Re$ has a unique representation of the form
\[f(x)=\sum_{S\subseteq[n]} \hat f(S)\,\chi_S(x),\] where $\hat
f(S)=\langle f,\chi_S\rangle$. The reals $\hat f(S)$ are
called the \emph{Fourier coefficients of $f.$}
The degree of $f,$ denoted $\deg(f),$ is the quantity 
$\max\{|S|:\hat f(S)\ne 0\}.$ The orthonormality
of $\{\chi_S\}$ immediately yields \emph{Parseval's identity}:
\begin{align}
  \sum_{S\subseteq[n]} \hat f(S)^2 
    = \langle f,f\rangle = \Exp_x[f(x)^2]. \label{eqn:parsevals}
\end{align}
The following fact is immediate from the definition of $\hat f(S)$:
\begin{proposition}
Let $f\colon \zoon\to\Re$ be given. Then
\[ \max_{S\subseteq[n]} |\hat f(S)| 
   \leq 2^{-n} \sum_{x\in\zoon} |f(x)|.\]
 \label{prop:fourier-coeff-bound}
\end{proposition}

A Boolean function $f\colon\zoon\to\moo$ is called \emph{symmetric}
if $f(x)$ is uniquely determined
by $\sum x_i.$ Equivalently, a Boolean function $f$ is symmetric
if and only if
\[ f(x_1,x_2,\dots,x_n) =
f(x_{\sigma(1)},x_{\sigma(2)},\dots,x_{\sigma(n)}) \]
for all inputs $x\in\zoon$ and all permutations $\sigma\colon[n]\to[n].$
Note that there is a one-to-one correspondence between predicates
and symmetric Boolean functions. Namely, one associates a predicate
$D$ with the symmetric function $f(x)\equiv D(\sum x_i).$

\subsection{Matrix analysis} \label{sec:matanal}
We draw freely on basic notions from matrix analysis. In particular,
we assume familiarity with the singular value decomposition; positive
semidefinite matrices; matrix similarity; matrix trace and its
properties; the Kronecker product and its spectral properties; the
relation between singular values and eigenvalues; and eigenvalue
computation for matrices of simple form. An excellent reference on
the subject is~\cite{matanal}. The review below is limited to
notation and the more substantial results.

The symbol $\Re^{m\times n}$ refers to
the family of all $m\times n$ matrices with real entries.
We specify matrices by their generic entry, e.g., $A=[F(i,j)]_{i,j}.$
In most matrices that arise in this work, the
exact ordering of the columns (and rows) is irrelevant. In such cases
we describe a matrix by the notation $[F(i,j)]_{i\in I,\, j\in
J},$ where $I$ and $J$ are some index sets. We denote the rank of
$A\in\Re^{m\times n}$ by $\rk A.$ We also write
\[ \|A\|_\infty   =    \max_{i,j} \;|A_{ij}|,
\qquad\qquad
   \|A\|_1   =    \sum_{i,j}|A_{ij}|.
\]
We denote the singular values of $A$ by
$\sigma_1(A)\geq\sigma_2(A)\geq\cdots\geq\sigma_{\min\{m,n\}}(A)\geq0.$
Recall that the spectral norm, trace norm, and Frobenius norm of
$A$ are given by
\begin{align*}
\|A\|
\phantom{\F{_\Sigma}}
&= \max_{x\in\Re^n,\; \|x\|=1} \|Ax\| = \sigma_1(A),\\
\|A\|_\Sigma
\phantom{\F}
&= \sum \sigma_i(A),\\
\|A\|\F
\phantom{_\Sigma}
&= \sqrt{\sum A_{ij}^2} = \sqrt{\sum \sigma_i(A)^2}.
\end{align*}
For a square matrix $A\in\Re^{n\times n},$ its trace is given by $\trace
A=\sum A_{ii}.$

Recall that every matrix $A\in\Re^{m\times n}$ has a singular
value decomposition $A=U\Sigma V\tr,$ where $U$ and $V$ are orthogonal
matrices and $\Sigma$ is diagonal with entries
$\sigma_1(A),\sigma_2(A),\dots,\sigma_{\min\{m,n\}}(A).$ For
$A,B\in\Re^{m\times n},$ we write $\langle A, B\rangle =
\sum A_{ij}B_{ij}=\trace(AB\tr).$ 
A~useful consequence of the singular value decomposition is:
\begin{equation}
\langle A,B\rangle \leq \|A\|\;\|B\|_\Sigma
\qquad\qquad (A,B\in\Re^{m\times n}).
\label{eqn:hoffman-wielandt-conseq}
\end{equation}

Following~\cite{razborov03quantum}, we define the
\emph{$\epsilon$-approximate trace norm} of a matrix $F\in\Re^{m\times
n}$ by
\begin{align*}
\|F\|_{\Sigma,\epsilon} = \min\{ \|A\|_\Sigma : \|F-A\|_\infty\leq
\epsilon\}.
\end{align*}
The next proposition is a trivial consequence of
(\ref{eqn:hoffman-wielandt-conseq}).
\begin{proposition}
Let $F\in\Re^{m\times n}$ and $\epsilon\geq0.$ Then
\begin{align*}
\|F\|_{\Sigma,\epsilon} \geq \sup_{\Psi\ne 0} 
\,\frac{\langle F,\Psi\rangle - \epsilon \|\Psi\|_1}
{\|\Psi\|}.
\end{align*}
\label{prop:approximate-trace-norm}
\end{proposition}
\begin{proof}
Fix any $\Psi\ne0$ and $A$ such that $\|F-A\|_\infty\leq \epsilon.$
Then $\langle A,\Psi\rangle \leq\|A\|_\Sigma \|\Psi\|$ by
(\ref{eqn:hoffman-wielandt-conseq}).  On the other hand,
$\langle A,\Psi\rangle
   \geq \langle F,\Psi\rangle - \|A-F\|_\infty\|\Psi\|_1 
   \geq \langle F,\Psi\rangle - \epsilon\|\Psi\|_1.
$
Comparing these two estimates of $\langle A,\Psi\rangle$ gives the sought
lower bound on $\|A\|_\Sigma.$
\qquad
\end{proof}

Following~\cite{buhrman-dewolf01polynomials},
we define the \emph{$\epsilon$-approximate rank}
of a matrix $F\in \Re^{m\times n}$ by
\begin{align*}
\rk_\epsilon F = \min\{ \rk A : \|F-A\|_\infty\leq
\epsilon\}.
\end{align*}
The approximate rank and approximate trace norm are related by virtue of
the singular value decomposition, as follows.
\begin{proposition}
Let $F\in\Re^{m\times n}$ and $\epsilon\geq 0$ be given. Then
\begin{align*}
\rk_\epsilon F \geq 
\frac{(\|F\|_{\Sigma,\epsilon})^2}
     {\sum_{i,j} (|F_{ij}|+\epsilon)^2}.
\end{align*}
\label{prop:approx-rank-approx-trace-norm}
\end{proposition}

\begin{annotatedproof}{Proof 
\textup{(adapted from~\cite{colt07rankeps})}}
Fix $A$ with $\|F-A\|_\infty\leq\epsilon.$ Then
\begin{align*}
   \qquad\qquad
  \|F\|_{\Sigma,\epsilon}\leq \|A\|_\Sigma
   \leq \|A\|\F \sqrt{\rk A} 
   \leq \left(\sum_{i,j} (|F_{ij}| + \epsilon)^2\right)^{1/2}\sqrt{\rk A}.
   \qquad\qquad\QED
\end{align*}
\end{annotatedproof}


We will also need a well-known bound on the trace norm of a
matrix product, which we state with a proof for the reader's
convenience.

\begin{proposition}
For all real matrices $A$ and $B$ of compatible dimensions,
\begin{align*}
\|AB\|_\Sigma \leq \|A\|\F \; \|B\|\F.
\end{align*}
\label{prop:bound-on-trace-of-product}
\end{proposition}
\begin{qedproof}
Write the singular value decomposition $AB=U\Sigma V\tr.$ Let
$u_1,u_2,\dots$ and $v_1,v_2,\dots$ stand for the columns of $U$
and $V,$ respectively.  By definition, $\|AB\|_\Sigma$ is the sum
of the diagonal entries of $\Sigma.$ We have:
\begin{align*}
\|AB\|_\Sigma &= \sum (U\tr ABV)_{ii} = \sum (u_i\tr A)(Bv_i)
 \leq \sum \|A\tr u_i\|\; \|Bv_i \| \\
 &\leq \sqrt{\sum \|A\tr u_i\|^2}\sqrt{\sum \|Bv_i \|^2 }
 = \|U\tr A\|\F \; \|BV\|\F  = \|A\|\F \; \|B\|\F. \qquad\QED
\end{align*}
\end{qedproof}

\subsection{Approximation and sign-representation} \label{sec:approx}
For a function $f\colon\zoon\to\Re,$ we define
\begin{align*}
E(f,d) = \min_p \|f - p\|_\infty,
\end{align*}
where the minimum is over real polynomials of degree up to $d.$ The
\emph{$\epsilon$-approximate degree} of $f,$ denoted $\degeps(f),$
is the least $d$ with $E(f,d)\leq\epsilon.$ In words, the
$\epsilon$-approximate degree of $f$ is the least degree of a
polynomial that approximates $f$ uniformly within $\epsilon.$

For a Boolean function $f\colon\zoon\to\moo,$ the $\epsilon$-approximate
degree is of particular interest for $\epsilon=1/3.$ The choice of
$\epsilon=1/3$ is a convention and can be replaced by any other
constant in $(0,1),$ without affecting $\degeps(f)$ by more than a
multiplicative constant.  Another well-studied notion is the
\emph{threshold degree} $\degthr(f),$ defined for a Boolean function
$f\colon\zoon\to\moo$ as the least degree of a real polynomial $p$ with 
$f(x)\equiv \sign p(x).$ In words, $\degthr(f)$ is the least degree of a
polynomial that represents $f$ in sign.

So far we have considered representations of Boolean
functions by \emph{real} polynomials. Restricting the polynomials
to have \emph{integer} coefficients yields another heavily studied
representation scheme. The main complexity measure here is the sum
of the absolute values of the coefficients.  Specifically, for a
Boolean function $f\colon\zoon\to\moo,$ its \emph{degree-$d$ threshold
weight} $W(f,d)$ is defined to be the minimum $\sum_{|S|\leq d}
|\lambda_S|$ over all integers $\lambda_S$ such that
\begin{align*}
f(x)\equiv \sign\left( \sum_{S\subseteq\{1,\dots,n\}, \, |S|\leq d}
\lambda_S\chi_S(x)\right).
\end{align*}
If no such integers $\lambda_S$ can be found, we put $W(f,d)=\infty.$ It is
straightforward to verify that the following three conditions are
equivalent: $W(f,d)=\infty$; $E(f,d)=1$; $d<\degthr(f).$
In all expressions involving $W(f,d),$ we adopt 
the standard convention that $1/\infty=0$ and $\min\{t,\infty\} =
t$ for any real $t.$

As one might expect, representations of Boolean functions by real and
integer polynomials are closely related. In particular, we have the
following relationship between $E(f,d)$ and $W(f,d).$

\begin{theorem}
Let $f\colon \zoon\to\moo$ be given. Then for $d=0,1,\dots,n,$
\begin{align*}
\frac1{1-E(f,d)} \leq W(f,d) \leq \frac2{1-E(f,d)}
\left\{
{n\choose 0}+{n\choose 1}+\cdots+{n\choose d}\right\}^{3/2},
\end{align*}
with the convention that $1/0=\infty.$
\label{thm:E-vs-W}
\end{theorem}

Since Theorem~\ref{thm:E-vs-W} is not directly used in our derivations,
we defer its proof to Appendix~\ref{sec:E-vs-W}. Similar
statements have been noted earlier by several authors~\cite{KP98threshold,
buhrman07pp-upp}.  We close this section with Paturi's tight
estimate~\cite{paturi92approx} of the approximate degree for each
symmetric Boolean function.

\begin{theorem}[Paturi]
Let $f\colon\zoon\to\moo$ be a given function such that $f(x)\equiv
D(\sum x_i)$ for some predicate $D\colon\zodn\to\moo.$ Then
\begin{align*}
\adeg (f) =\Theta\left(\sqrt{ n\ell_0(f)} + \sqrt{n\ell_1(f)}\right),
\end{align*}
where
$\ell_0(D)\in\{0,1,\dots,\lfloor n/2\rfloor\}$ and
$\ell_1(D)\in\{0,1,\dots,\lceil n/2\rceil\}$
are the smallest integers such that $D$ is constant in the range
$[\ell_0(D),n-\ell_1(D)].$ 
\label{thm:paturi}
\end{theorem}


\subsection{Quantum communication} \label{sec:quantum-model}

This section reviews the quantum model of communication complexity.
We include this review mainly for completeness; our proofs rely solely on
a basic matrix-analytic property of such protocols and on no other
aspect of quantum communication.

There are several equivalent ways to describe a quantum communication
protocol. Our description closely follows
Razborov~\cite{razborov03quantum}.  Let $\AA$ and $\BB$ be complex
finite-dimensional Hilbert spaces.  Let $\CC$ be a Hilbert space
of dimension $2,$ whose orthonormal basis we denote by
$|0\rangle,\;|1\rangle.$  Consider the tensor product
$\AA\otimes\CC\otimes\BB,$ which is itself a Hilbert space with an
inner product inherited from $\AA,$ $\BB,$ and $\CC.$ The \emph{state}
of a quantum system is a unit vector in $\AA\otimes\CC\otimes\BB,$
and conversely any such unit vector corresponds to a distinct quantum
state.  The quantum system starts in a given state and traverses a
sequence of states, each obtained from the previous one via a unitary
transformation chosen according to the protocol.  Formally, a
\emph{quantum communication protocol} is a finite sequence of unitary
transformations
\[U_1\otimes I_\BB,\quad
 I_\AA\otimes U_2,\quad
 U_3\otimes I_\BB,\quad
 I_\AA\otimes U_4,\quad
 \dots,\quad
U_{2k-1}\otimes I_\BB,\quad
I_\AA\otimes U_{2k},\]
where: $I_\AA$ and $I_\BB$ are the identity transformations in $\AA$
and $\BB,$ respectively; $U_1,U_3,\dots,U_{2k-1}$ are unitary
transformations in $\AA\otimes\CC$; and $U_2,U_4,\dots,U_{2k}$ are
unitary transformations in $\CC\otimes\BB.$ The \emph{cost} of the
protocol is the length of this sequence, namely, $2k.$ On Alice's
input $x\in X$ and Bob's input $y\in Y$ (where $X,Y$ are given finite sets),
the computation proceeds as follows.

\begin{enumerate}
\item The quantum system starts out in an initial state $\init(x,y).$
\item Through successive applications of the above unitary
transformations, the system reaches the state
\[
 \final(x,y)   =    (I_\AA\otimes U_{2k})
 (U_{2k-1}\otimes I_\BB)
 \cdots
 (I_\AA\otimes U_2)
 (U_1\otimes I_\BB)\;\init(x,y). \]
\item Let $v$ denote the projection of $\final(x,y)$ onto
$\AA\otimes\Span(|1\rangle) \otimes\BB.$ The output of the protocol
is $1$ with probability $\langle v,v\rangle,$ and $0$ with the
complementary probability
$1-\langle v,v\rangle.$
\end{enumerate}

All that remains is to specify how the initial state
$\init(x,y)\in\AA\otimes\CC\otimes\BB$ is constructed from
$x,y.$ It is here that the model with prior entanglement
differs from the model without prior entanglement.
In the model without prior entanglement, $\AA$  and $\BB$ have
orthonormal bases 
$\{|x,w\rangle: x\in X,\; w\in W \}$ and 
$\{|y,w\rangle: y\in Y,\; w\in W \},$ 
respectively, where $W$ is a finite set corresponding to the private
workspace of each of the parties. The initial state is the pure
state
\[ 
\init(x,y) = |x,0\rangle\,|0\rangle\,|y,0\rangle, 
\]
where $0\in W$ is a certain fixed element.  In the model with prior
entanglement, the spaces $\AA$ and $\BB$ have orthonormal bases
$\{|x,w,e\rangle:x\in X,\;w\in W,\; e\in E\}$ and 
$\{|y,w,e\rangle:y\in Y,\;w\in W,\; e\in E\},$
respectively, where $W$ is as before and $E$ is a finite set
corresponding to the prior entanglement.  The initial state is now 
the entangled state
\[ \init(x,y) =
\frac{1}{\sqrt{|E|}} \sum_{e\in E} 
|x,0,e\rangle\,|0\rangle\,|y,0,e\rangle.
\]
Apart from finite size, no assumptions are made about $W$ or $E.$
In particular, the model with prior entanglement allows for an
unlimited supply of entangled qubits.  This mirrors the unlimited
supply of shared random bits in the classical public-coin randomized
model.

Let $f\colon X\times Y\to\moo$ be a given function. A quantum
protocol $P$ is said to compute $f$ with error $\epsilon$ if
\[
\Prob\left[f(x,y)=(-1)^{P(x,y)}\right]\geq 1-\epsilon 
\]
for all $x,y,$
where the random variable $P(x,y)\in\zoo$ is the output of the
protocol on input $(x,y).$ Let $Q_\epsilon(f)$ denote the least cost of
a quantum protocol without prior entanglement that computes
$f$ with error $\epsilon.$ Define $Q^*_\epsilon(f)$ analogously for
protocols with prior entanglement.  The precise choice of a constant
$0<\epsilon<1/2$ affects $Q_\epsilon(f)$ and $Q^*_\epsilon(f)$
by at most a constant factor, and thus the setting $\epsilon=1/3$
entails no loss of generality.

Let $D\colon \zodn\to\moo$ be a predicate. We associate with
$D$ the function $f\colon \zoon\times\zoon\to\moo$ defined by
$f(x,y)    =    D(\sum x_iy_i).$
We let $Q_\epsilon(D) = Q_\epsilon(f)$ and $Q^*_\epsilon(D) =
Q^*_\epsilon(f).$ More generally, by computing $D$ in the quantum
model we mean computing the associated function $f.$ We write
$R_\epsilon(f)$ for the least cost of a classical public-coin
protocol for $f$ that errs with probability at most $\epsilon$ on
any given input. Another classical model that figures in this paper
is the \emph{deterministic} model. We let $D(f)$ denote the
deterministic communication complexity of $f.$ Throughout this
paper, by the communication complexity of a Boolean matrix
$F=[F_{ij}]_{i\in I,\,j\in J}$ we will mean the communication
complexity of the associated function $f\colon I\times J\to\moo,$
given by $  f(i,j) = F_{ij}. $

\subsection{The generalized discrepancy method} \label{sec:discrepancy}

The generalized discrepancy method is an intuitive and elegant technique for
proving communication lower bounds.  A starting point
in our discussion is the following fact due to 
Linial and Shraibman~\cite[Lem.~10]{linial07factorization-stoc}, with
closely analogous statements established earlier by Yao~\cite{yao93quantum}, 
Kremer~\cite{kremer95thesis}, and Razborov~\cite{razborov03quantum}.

\begin{theorem}
Let $X, Y$ be finite sets.  Let $P$ be a quantum protocol $($with or
without prior entanglement$)$ with cost $C$ qubits and input sets $X$
and $Y.$ Then
\[ \Big[\Exp[P(x,y)]\Big]_{x,y} = AB \]
for some real matrices $A, B$ with $\|A\|\F \leq 2^C\sqrt{|X|}$
\;and\;
$\|B\|\F \leq 2^C\sqrt{|Y|}.$
\label{thm:protocol2matrix}
\end{theorem}

Theorem~\ref{thm:protocol2matrix} states that the matrix of acceptance
probabilities of every low-cost protocol
$P$ has a nontrivial factorization. This transition from quantum
protocols to matrix factorization is a standard technique and has
been used by various authors in various contexts.

The generalized discrepancy method was
first applied by Klauck~\cite[Thm.~4]{klauck01quantum}
and reformulated more broadly by Razborov~\cite{razborov03quantum}.
The treatment in~\cite{razborov03quantum} is informal.
In what follows, we propose a precise formulation of the
generalized discrepancy method and supply a proof.

\begin{theorem}[generalized discrepancy method]
Let $X,Y$ be finite sets and $f\colon X\times Y\to\moo$ a given function.
Let $\Psi=[\Psi_{xy}]_{x\in X,\,y\in Y}$ be any real matrix with $\|\Psi\|_1=1.$
Then for each $\epsilon>0,$
\[ 4^{Q_\epsilon(f)} \geq 4^{Q^*_\epsilon(f)} \geq
\frac{\langle \Psi, F\rangle - 2\epsilon}{3\,\|\Psi\|\sqrt{|X|\,|Y|}}, \]
where $F=[f(x,y)]_{x\in X,\,y\in Y}.$
\label{thm:discrepancy-method}
\end{theorem}

\begin{proof}
Let $P$ be a quantum protocol with prior entanglement that computes
$f$ with error $\epsilon$ and cost $C.$ Put
\[ \Pi    =    \Big[ \Exp[P(x,y)] \Big]_{x\in X,\,y\in Y}. \]
Then  we can write $F = (J-2\Pi) + 2E,$ where $J$ is the all-ones
matrix and $E$ is some matrix with $\|E\|_\infty \leq \epsilon.$
As a result,
\begin{align}
\langle \Psi,J-2\Pi\rangle 
&=    \langle \Psi, F\rangle - 2\,\langle \Psi,E\rangle \nonumber\\
&\geq \langle \Psi, F\rangle - 2\epsilon\,\|\Psi\|_1 \nonumber\\
&=  \langle \Psi, F\rangle - 2\epsilon.
\label{eqn:lower-inner-product}
\end{align}
	On the other hand, Theorem~\ref{thm:protocol2matrix}
	guarantees the existence of matrices $A$ and $B$ with $AB=\Pi$
	and $\|A\|\F\,\|B\|\F\leq 4^C\sqrt{|X|\,|Y|}.$ Therefore,
\begin{align}
\langle \Psi,J-2\Pi\rangle 
  &\leq \|\Psi\|\; \|J-2\Pi\|_\Sigma  
     &&\text{by (\ref{eqn:hoffman-wielandt-conseq})} \nonumber\\
  &\leq \|\Psi\|\; \left(\!\sqrt{|X|\,|Y|} + 2\,\|\Pi\|_\Sigma \right) 
     &&\text{since $\|J\|_\Sigma=\sqrt{|X|\,|Y|}$} \nonumber\\
  &\leq \|\Psi\|\; \left(\!\sqrt{|X|\,|Y|} +  2\,\|A\|\F\;\|B\|\F \right)   
     &&\text{by Prop.~\ref{prop:bound-on-trace-of-product}} 
	 \nonumber \\
  &\leq \|\Psi\|\; \left(2\cdot 4^{C}+1\right)\sqrt{|X|\;|Y|}.
\label{eqn:upper-inner-product}
\end{align}
The theorem follows by comparing (\ref{eqn:lower-inner-product}) and
(\ref{eqn:upper-inner-product}). \qquad
\end{proof}

\begin{remark}{\rm
Theorem~\ref{thm:discrepancy-method} is not to be
confused with Razborov's \emph{multidimensional} technique, also
found in~\cite{razborov03quantum}, which we will have no
occasion to use or describe. 
}
\end{remark}


We will now abstract away the particulars of 
Theorem~\ref{thm:discrepancy-method} 
and articulate the fundamental mathematical technique in
question.
This will clarify the generalized discrepancy method and show
that it is simply an extension of Yao's original 
discrepancy method~\cite[\S3.5]{ccbook}.  
Let $f\colon X\times Y\to\moo$ be a given function
whose communication complexity we wish to estimate. The underlying
communication model is \emph{irrelevant} at this point.  Suppose we can
find a function $h\colon X\times Y\to\moo$ and a distribution $\mu$ on
$X\times Y$ that satisfy the following two properties.
\begin{enumerate}
	\item \emph{Correlation.}
	The functions $f$ and $h$ are well correlated under $\mu$:
	\begin{equation}
	\Exp_{(x,y)\sim\mu}\left[f(x,y)h(x,y)\right] 
			\geq \epsilon,
	\label{eqn:correlation-f-g}
	\end{equation}
	where $\epsilon>0$ is a given constant.

	\item \emph{Hardness.}
	No low-cost protocol $P$ in the given model of communication
	can compute $h$ to a substantial advantage under $\mu.$
	Formally, if $P\colon X\times Y\to\zoo$ 
	is a protocol in the given model with cost
	$C$ bits, then
	\begin{equation}
	\Exp_{(x,y)\sim\mu}\left[h(x,y)\Exp\left[(-1)^{P(x,y)}\right]\right]
	\leq 2^{O(C)}\gamma,
	\label{eqn:hardness-under-mu}
	\end{equation}
	where $\gamma=o(1).$ The inner expectation in
	(\ref{eqn:hardness-under-mu}) is over the internal operation
	of the protocol on the fixed input $(x,y).$
\end{enumerate}

If the above two conditions hold, we claim that any protocol in the
given model that computes $f$ with error at most $\epsilon/3$ on
each input must have cost $\Omega(\log\{\epsilon/\gamma\}).$ 
Indeed, let $P$ be a protocol with
$\Prob[P(x,y)\ne f(x,y)]\leq \epsilon/3$ for all $x,y.$ Then standard
manipulations reveal:
\begin{align*}
   \Exp_{\mu}
       \left[h(x,y)\Exp\left[(-1)^{P(x,y)}\right]\right]
 \geq \Exp_{\mu}\left[f(x,y)h(x,y)\right] - 2\cdot
        \frac{\epsilon}{3}
  \geq 
         \frac{\epsilon}3,
\end{align*}
where the last step uses (\ref{eqn:correlation-f-g}).
In view of (\ref{eqn:hardness-under-mu}), this shows that $P$ must
have cost $\Omega(\log \{\epsilon/\gamma\}).$

We attach the term \emph{generalized discrepancy method} to this
abstract framework.  Readers with background in communication
complexity will note that the original discrepancy method of
Yao~\cite[\S3.5]{ccbook} corresponds to the case when $f=h$ and the
communication takes place in the two-party randomized model.

The purpose of our abstract discussion was to expose the
fundamental mathematical technique in question, which is 
independent of the communication model.
Indeed, the communication model enters the picture only in the proof
of~(\ref{eqn:hardness-under-mu}). It is here that the analysis must
exploit the particularities of the model. To place an upper bound
on the advantage under $\mu$ in the quantum model with entanglement,
as we see from~(\ref{eqn:upper-inner-product}), 
one considers the quantity $\|\Psi\|\sqrt{|X|\,|Y|},$ where
$\Psi=[h(x,y)\mu(x,y)]_{x,y}.$ In the classical randomized model, the
quantity to estimate happens to be
\[ 
\max_{\rule{0mm}{4mm}\substack{S\subseteq X,\\T\subseteq Y}}
\quad
	\left| \sum_{x\in S} \sum_{y\in T} \mu(x,y) h(x,y)\right|
	, \]
which is known as the \emph{discrepancy} of $h$ under $\mu.$


\section{Duals of approximation and sign-representation} \label{sec:berr-morth}
Crucial to our work are the dual characterizations of the uniform approximation
and sign-representation of Boolean functions by real polynomials.
As a starting point, we recall a classical result from approximation
theory due to Ioffe and Tikhomirov~\cite{ioffe-tikhomirov68duality}
on the duality of norms. A more recent treatment is available in 
the textbook of DeVore and Lorentz~\cite{devore-lorentz93approx-book},
p.~61,~Thm.~1.3. We provide a short and elementary proof of this
result in Euclidean space, which will suffice for our purposes. 
We let $\Re^X$ stand for the linear space of real
functions on the set $X.$

\begin{theorem}[Ioffe and Tikhomirov]
Let $X$ be a finite set. Fix $\Phi \subseteq\Re^X$ and a function
$f\colon X\to\Re.$ Then
\begin{align}
 \min_{\phi\in\Span(\Phi)} \|f - \phi\|_\infty = 
  \max_\psi \left\{ \sum_{x\in X} f(x)\psi(x) \right\},
\label{eqn:berr-morth}
\end{align}
where the maximum is over all functions $\psi\colon X\to\Re$ such that
\begin{align*}
\sum_{x\in X} |\psi(x)| \leq 1
\end{align*}
and, for each $\phi\in\Phi,$
\begin{align*}
\sum_{x\in X}\phi(x)\psi(x)=0.
\end{align*}
\label{thm:berr-morth}
\end{theorem}

\begin{proof}
The theorem holds trivially when $\Span(\Phi)=\{0\}.$ Otherwise,
let $\phi_1,\dots,\phi_k$ be a basis for $\Span(\Phi).$ 
Observe that the left member of (\ref{eqn:berr-morth})
is the optimum of the following linear program in the variables
$\epsilon,\alpha_1,\dots,\alpha_k$:
\begin{align*}
\framebox{\parbox{0.75\textwidth}{
		$\begin{aligned}
		\text{minimize:} \quad & \epsilon \hspace{2cm}\\
		\text{subject to:} \quad &
			   \left| f(x)  -\sum_{i=1}^k\alpha_i\phi_i(x)\right|\leq
				 \epsilon            &&\quad \text{for each } x\in X, \\
			  &\alpha_i\in \Re       &&\quad \text{for each }i,\\
			  &\epsilon\geq 0.\\
		\end{aligned}$
}}
\end{align*}
Standard manipulations reveal the dual:
\begin{align*}
\framebox{\parbox{0.75\textwidth}{
		$\begin{aligned}
		\text{maximize:} \quad & \sum_{x\in X}\psi_x f(x) \\
		\text{subject to:} \quad 
			   &\sum_{x\in X}|\psi_x| \leq 1, \\
			  &\sum_{x\in X}\psi_x\phi_i(x) = 0
					  &&\quad \text{for each }i, \\
			  &\psi_x \in \Re &&\quad\text{for each }x\in X.
		\end{aligned}$
}}
\end{align*}

\noindent
Both programs are clearly feasible and thus have the same finite
optimum.  We have already observed that the optimum of first program
is the left-hand side of (\ref{eqn:berr-morth}). 
Since $\phi_1,\dots,\phi_k$ form a basis
for $\Span(\Phi),$ the optimum of the second program is by definition
the right-hand side of (\ref{eqn:berr-morth}). \qquad
\end{proof}

As a corollary to Theorem~\ref{thm:berr-morth}, we obtain a dual
characterization of the approximate degree.

\newpage

\begin{theorem}[approximate degree]
Fix $\epsilon\geq0.$ Let $f\colon\zoon\to\Re$ be given, $d=
\degeps(f)\geq1.$ Then there is a function $\psi\colon\zoon\to\Re$
such that
\begin{align*}
&\;\,\hat\psi(S)=0 &&  (|S|<d),\\
&\sum_{x\in\zoon}|\psi(x)| =1, \\
&\sum_{x\in\zoon}\psi(x)f(x) > \epsilon.
\end{align*}
\label{thm:dual-approx}
\end{theorem}

\begin{proof}
Set $X=\zoon$ and $\Phi = \{\chi_S: |S|< d \}\subset\Re^X.$ Since
$\degeps(f)=d,$ we conclude that 
\begin{align*}
\min_{\phi\in \Span(\Phi)}\|f-\phi\|_\infty >\epsilon.
\end{align*}
In view of Theorem~\ref{thm:berr-morth},  we can take $\psi$ to be
any function for which the maximum is achieved in~(\ref{eqn:berr-morth}).
\qquad
\end{proof}

We now state the dual characterization of the threshold degree, which
is better known as Gordan's Transposition Theorem~\cite[\S7.8]{lpbook}.

\begin{theorem}[threshold degree]
Let $f\colon\zoon\to\moo$ be given, $d=\degthr(f).$ Then
there is a distribution $\mu$ over $\zoon$ with
\begin{align*}
  \Exp_{x\sim\mu}[f(x)\chi_S(x)]=0 &&(|S|<d).
\end{align*}
\label{thm:thresh-or-distr}
\end{theorem}

See~\cite{sherstov07ac-majmaj} for a derivation of
Theorem~\ref{thm:thresh-or-distr} using linear programming duality.
Alternately, it can be derived as a corollary to
Theorem~\ref{thm:berr-morth}.  We close this section with one final
dual characterization, corresponding to sign-representation by
integer polynomials.

\begin{theorem}[threshold weight]
Fix a function $f\colon\zoon\to\moo$ and an integer 
$d\geq\degthr(f).$ Then for every distribution $\mu$ on $\zoon,$
\begin{align}
\max_{|S|\leq d} \left\lvert \Exp_{x\sim\mu}[f(x)\chi_S(x)]\right\rvert
	&\geq \frac{1}{W(f,d)}.
\label{eqn:weight-lower}
\intertext{Furthermore, there exists a distribution $\mu$ such that}
\max_{|S|\leq d} \left\lvert \Exp_{x\sim\mu}[f(x)\chi_S(x)]\right\rvert
	&\leq \left(\frac{2n}{W(f,d)}\right)^{1/2}.
\label{eqn:weight-upper}
\end{align}
\label{thm:dual-weight}
\end{theorem}

Inequalities (\ref{eqn:weight-lower}) and (\ref{eqn:weight-upper})
are originally due to Hajnal et al.~\cite{hajnal93threshold-const-depth} and
Freund~\cite{freund90boosting}, respectively.  For an integrated
treatment of both results, see Goldmann et al.~\cite{GHR92}, Lem.~4
and Thm.~10.

\section{Pattern matrices} \label{sec:pattern-matrices}

We now turn to the second ingredient of our proof, a certain family
of real matrices that we introduce. Our goal
here is to explicitly calculate their singular values. As we shall
see later, this provides a convenient means to generate hard
communication problems. 

Let $t$ and $n$ be positive integers, where $t<n$ and $t\mid n.$ Partition
$[n]$ into $t$ contiguous blocks, each with $n/t$ elements:
\[
[n] = \left\{1,  2,  \dots,   \frac{n}{t}\right\} 
       \cup 
	   \left\{\frac{n}{t}+1,   \dots,  \frac{2n}{t}\right\}
         \cup \cdots \cup
         \left\{\frac{(t-1)n}{t}+1,  \dots,  n\right\}.
\]
Let $\VV(n,t)$ denote the family of subsets $V\subseteq[n]$
that have exactly one element in each of these blocks (in particular,
$|V|=t$).  Clearly, $|\VV(n,t)|=(n/t)^t.$ For a bit string $x\in\zoon$
and a set $V\in\VV(n,t),$ define the \emph{projection of $x$ onto
$V$} by
\[ x|_V  =      (x_{i_1},x_{i_2},\dots,x_{i_t}) \in  \zoo^t, \]
where $i_1<i_2<\cdots< i_t$ are the elements of $V.$ We are ready for a
formal definition of our matrix family.

\begin{definition}[pattern matrix]
{\rm 
For $\phi\colon \zoo^t\to\Re,$ the \emph{$(n,t,\phi)$-pattern matrix} is
the real matrix $A$ given by
\[ A = \Big[\phi(x|_V\oplus
w)\Big]_{x\in\zoon,\,(V,w)\in\VV(n,t)\times\zoo^t}  \;. \]
In words, $A$ is the matrix of size $2^n$~by~$(n/t)^t2^t$ whose rows are
indexed by strings $x\in\zoon,$ whose columns are indexed by pairs
$(V,w)\in\VV(n,t)\times\zoo^t,$ and whose entries are given by
$A_{x,(V,w)}= \phi(x|_V\oplus w).$
}
\end{definition}

The logic behind the term ``pattern matrix" is as follows: a mosaic
arises from repetitions of a pattern in the same way that $A$ arises
from applications of $\phi$ to various subsets of the variables.
Our approach to analyzing the singular values of a pattern matrix
$A$ will be to represent it as the sum of simpler matrices and
analyze them instead.  For this to work, we should be able to
reconstruct the singular values of $A$ from those of the simpler
matrices. Just when this can be done is the subject of the following
lemma.

\begin{lemma}[singular values of a matrix sum]
Let $A,B$ be real matrices with $AB\tr=0$ and $A\tr B=0.$
Then the nonzero singular values of $A+B,$ counting
multiplicities, are $\sigma_1(A),\dots,\sigma_{\rk
A}(A),\sigma_1(B),\dots,\sigma_{\rk B}(B).$
\label{lem:disjoint-spectra}
\end{lemma}

\begin{proof}
The claim is trivial when $A=0$ or $B=0,$ so assume otherwise.
Since the singular values of $A+B$ are precisely the square roots of the
eigenvalues of $(A+B)(A+B)\tr,$ it suffices
to compute the spectrum of the latter matrix. Now,
\begin{align}
(A+B)(A+B)\tr 
&= AA\tr + BB\tr + \underbrace{AB\tr}_{=0} + \underbrace{BA\tr}_{=0} 
   \nonumber \\
&= AA\tr + BB\tr.   \label{eqn:aat-bbt}
\end{align}
Fix spectral decompositions 
\[ AA\tr = \sum_{i=1}^{\rk A}\sigma_i(A)^2 u_iu_i\tr, \qquad
   BB\tr = \sum_{j=1}^{\rk B}\sigma_j(B)^2 v_jv_j\tr. \]
Then
\begin{align}
\sum_{i=1}^{\rk A}\; \sum_{j=1}^{\rk B} 
    \sigma_i(A)^2\sigma_j(B)^2 
    \langle u_i,v_j\rangle^2 
	\;\;
&=\;\; \left\langle 
	\sum_{i=1}^{\rk A} \sigma_i(A)^2 u_iu_i\tr,
	\sum_{j=1}^{\rk B} \sigma_j(B)^2 v_jv_j\tr \right\rangle \nonumber\\
&=\;\; \langle AA\tr, BB\tr\rangle           \nonumber \\
&=\;\; \trace(AA\tr BB\tr)                   \nonumber \\
&=\;\; \trace(A\cdot 0\cdot B\tr)            \nonumber \\
&=\;\; 0.   \label{eqn:orthogonal-eigenvectors}
\end{align}
Since $\sigma_i(A)\,\sigma_j(B)>0$ for all $i,j,$ it follows from
(\ref{eqn:orthogonal-eigenvectors}) that $\langle u_i,v_j\rangle=0$
for all $i,j.$ Put differently, the vectors $u_1,\dots,u_{\rk
A},v_1,\dots,v_{\rk B}$ form an orthonormal set. Recalling
(\ref{eqn:aat-bbt}), we conclude that the spectral decomposition
of $(A+B)(A+B)\tr$ is
\[\sum_{i=1}^{\rk A}\sigma_i(A)^2 u_iu_i\tr + 
\sum_{j=1}^{\rk B}\sigma_j(B)^2 v_jv_j\tr, \]
and thus the nonzero eigenvalues of $(A+B)(A+B)\tr$ are as claimed.
\qquad
\end{proof}


We are ready for the main result of this section.

\begin{theorem}[singular values of a pattern matrix]
Let $\phi\colon \zoo^t\to\Re$ be given. Let $A$ be the $(n,t,\phi)$-pattern
matrix. Then the nonzero singular values of $A,$ counting 
multiplicities, are:
\[ \bigcup_{S:\hat\phi(S)\ne0} \left\{ 
\sqrt{ 2^{n+t} \left( \frac{n}{t}\right)^t} 
\cdot |\hat\phi(S)| \left( \frac{t}{n}\right)^{|S|/2},
\quad\quad \text{repeated } \left(\frac{n}{t}\right)^{|S|}
           \text{ times}
\right\}.
\]
In particular,
\[ 
\|A\| \;=\; \sqrt{ 2^{n+t} \left( \frac{n}{t}\right)^t} 
\;\max_{S\subseteq[t]} 
    \left\{ |\hat\phi(S)| \left( \frac{t}{n}\right)^{|S|/2} \right\}.
\]
\label{thm:pattern-spectrum}
\end{theorem}

\begin{proof}
For each $S\subseteq[t],$ let $A_S$ be the $(n,t,\chi_S)$-pattern
matrix.  Thus,
\begin{equation}
A = \sum_{S\subseteq[t]} \hat\phi(S) A_S.
\label{eqn:A-as-sum}
\end{equation}
Fix arbitrary $S,T\subseteq[t]$ with $S\ne T.$ Then
\begin{align}
A_SA_T\tr &= \left[ \sum_{V\in\VV(n,t)}\; \sum_{w\in\zoo^t} 
                    \chi_S(x|_V\;\oplus\; w)
					\;\chi_T(y|_V\;\oplus\; w)  
              \right]_{x,y}         \nonumber \\
            &= 
			\left[\rule{0mm}{8mm} \right.
			       \sum_{V\in\VV(n,t)}
                    \chi_S(x|_V)
					\;\chi_T(y|_V)  
                    \underbrace{\sum_{w\in\zoo^t}\chi_S(w)\;\chi_T(w)}_{=0}
              \left.\rule{0mm}{8mm}
			  \right]_{x,y} \nonumber\\
            &= 0.  \label{eqn:aat}
\end{align}
Similarly,
\begin{align}
A_S\tr A_T 
  = \left[\rule{0mm}{8mm}\right. 
     \chi_S(w)\;\chi_T(w') 
      \underbrace{\sum_{x\in\zoo^n} \chi_S(x|_V)\;\chi_T(x|_{V'})}_{=0}
  \left.\rule{0mm}{8mm}\right]_{(V,w),(V',w')}
  = 0. 
\label{eqn:ata}
\end{align}
By (\ref{eqn:A-as-sum})--(\ref{eqn:ata}) and
Lemma~\ref{lem:disjoint-spectra}, the nonzero singular values of
$A$ are the union of the nonzero singular values of all $\hat\phi(S)A_S,$
counting multiplicities. Therefore, the proof will be complete once
we show that the only nonzero singular value of $A_S\tr A_S$ is
$2^{n+t}(n/t)^{t-|S|},$ with multiplicity $(n/t)^{|S|}.$ 
It is convenient to write this matrix as the Kronecker product
\[ A_S\tr A_S  \;\;=\;\; [\chi_S(w)\chi_S(w')]_{w,w'}\;\otimes \;
\left[\sum_{x\in\zoon}\chi_S(x|_V)\;\chi_S(x|_{V'})\right]_{V,V'}.
\]
The first matrix in this factorization has rank~$1$ and entries
$\pm1,$ which means that its only nonzero singular value is $2^t$ with
multiplicity~$1.$ The other matrix, call it $M,$ is permutation-similar
to
\[ 2^n\begin{bmatrix}
J & & &  \\
& J & & \\
& & \ddots & \\
& & & J
\end{bmatrix},
\]
where $J$ is the all-ones square matrix of order $(n/t)^{t-|S|}.$
This means that the only nonzero singular value of $M$ is
$2^n(n/t)^{t-|S|}$ with multiplicity $(n/t)^{|S|}.$ It follows from
elementary properties of the Kronecker product that the spectrum
of $A_S\tr A_S$ is as claimed.
\qquad
\end{proof}

\section{Pattern matrix method using uniform approximation} 
\label{sec:pattern-matrix-method}

The previous two sections examined relevant dual representations and 
the spectrum of pattern matrices. Having
studied these notions in their pure and basic form, we now apply
our findings to communication complexity. Specifically, we establish
the \emph{pattern matrix method} for communication complexity,
which gives strong lower bounds for every pattern matrix generated
by a Boolean function with high approximate degree.

\begin{restatetheorem}{thm:main-cc}
Let $F$ be the $(n,t,f)$-pattern matrix, 
where $f\colon \zoo^t\to\moo$ is given.
Then for every $\epsilon\in[0,1)$ and every $\delta<\epsilon/2,$
\begin{align}
Q^*_{\delta}(F) &\geq
\frac{1}{4}  \degeps(f)\log \left(\frac{n}{t}\right) - 
\frac12 \log\left(\frac{3}{\epsilon-2\delta}\right).
\label{eqn:pattern-matrix-general-error}
\intertext{In particular,}
Q^*_{1/7}(F) &>
\frac{1}{4}  \deg_{1/3}(f)\log \left(\frac{n}{t}\right) - 3. 
\label{eqn:pattern-matrix-bounded-error}
\end{align}
\end{restatetheorem}

\begin{proof}
Since (\ref{eqn:pattern-matrix-general-error}) immediately implies
(\ref{eqn:pattern-matrix-bounded-error}), we will focus on the former
in the remainder of the proof.
Let $d=\degeps(f)\geq1.$ By Theorem~\ref{thm:dual-approx},
there is a function $\psi\colon \zoo^t\to\Re$ such that:
\begin{align}
&\;\,\hat\psi(S)=0 &&(|S|<d),
	\label{eqn:psi-fourier-coeffs} \\
&\sum_{z\in\zoo^t}|\psi(z)| =1,
	\label{eqn:psi-bounded} \\
&\sum_{z\in\zoo^t}\psi(z)f(z) > \epsilon.
	\label{eqn:psi-correl}
\end{align}
Let $\Psi$ be the
$(n,t,2^{-n}(n/t)^{-t}\psi)$-pattern matrix. Then 
(\ref{eqn:psi-bounded}) and (\ref{eqn:psi-correl}) show that
\begin{equation}
\|\Psi\|_1 =1,  \qquad    \langle F, \Psi\rangle > \epsilon.
\label{eqn:K-bounded,K-M-correl}
\end{equation}
Our last task is to calculate $\|\Psi\|.$ By (\ref{eqn:psi-bounded})
and Proposition~\ref{prop:fourier-coeff-bound},
\begin{equation}
\max_{S\subseteq[t]} |\hat\psi(S)| \leq 2^{-t}.
	\label{eqn:max-fourier-coeff-psi}
\end{equation}
Theorem~\ref{thm:pattern-spectrum} yields, in view of
(\ref{eqn:psi-fourier-coeffs}) and (\ref{eqn:max-fourier-coeff-psi}):
\begin{equation}
\|\Psi\|
\leq  \left(\frac{t}{n}\right)^{d/2} 
      \left(2^{n+t} \left(\frac{n}{t}\right)^t\right)^{-1/2}. 
\label{eqn:K-norm}
\end{equation}
Now (\ref{eqn:pattern-matrix-general-error})
follows from (\ref{eqn:K-bounded,K-M-correl}), 
(\ref{eqn:K-norm}), and Theorem~\ref{thm:discrepancy-method}.
\qquad
\end{proof}

Theorem~\ref{thm:main-cc} gives
lower bounds not only for bounded-error communication but also for
communication protocols with error probability $\frac12-o(1).$ 
For example, if a function $f\colon\zoo^t\to\moo$ requires a polynomial of
degree $d$ for approximation within $1-o(1),$ equation
(\ref{eqn:pattern-matrix-general-error}) gives a lower bound for 
small-bias communication. We will complement and refine
that estimate in the next
section, which is dedicated to small-bias communication.

We now prove the corollary to Theorem~\ref{thm:main-cc} on function
composition, stated in the introduction.

\begin{annotatedproof}{Proof of Corollary~\textup{\ref{cor:main}}}
The $(2t,t,f)$-pattern matrix occurs as a submatrix of
$[F(x,y)]_{x,y\in\zoo^{4t}}.$ 
\qquad\QED
\end{annotatedproof}

Finally, we show that the lower
bound~(\ref{eqn:pattern-matrix-bounded-error}) derived above for
bounded-error communication complexity is tight up to a polynomial factor,
even for deterministic protocols. The proof follows a well-known
argument in the literature~\cite{bcw98quantum,
beals-et-al01quantum-by-polynomials}, as pointed out to us by R.~de
Wolf~\cite{dewolf-personal-oct-2007}.

\begin{proposition}[on the tightness of
Theorem~\ref{thm:main-cc}]
Let $F$ be the $(n,t,f)$-pattern matrix, 
where $f\colon \zoo^t\to\moo$ is given. Then
\begin{align*}
D(F) \leq O(\dt(f)\log(n/t)) \leq O(\deg_{1/3}(f)^6\log (n/t)),
\end{align*}
where $\dt(f)$ is the least depth of a decision tree for $f.$
In particular, \textup{(\ref{eqn:pattern-matrix-bounded-error})}
is tight up to a polynomial factor.
\label{prop:det-upper-bound}
\end{proposition}

\begin{proof}
Beals al.~\cite[Cor.~5.6]{beals-et-al01quantum-by-polynomials}
prove that $\dt(f)\leq O(\deg_{1/3}(f)^6)$ for all Boolean functions $f.$
Therefore, it suffices to prove an upper bound of $O(d\log(n/t))$ on the
deterministic communication complexity of $F,$ where $d=\dt(f).$

The needed deterministic protocol is well-known.
Fix a depth-$d$ decision tree for $f.$ 
Let $(x,(V,w))$ be a given input. Alice and Bob start at the root of the
decision tree, labeled by some variable $i\in\{1,\dots,t\}.$ By exchanging
$\lceil\log(n/t)\rceil+2$ bits, Alice and Bob determine $(x|_V)_i\oplus
w_i\in\zoo$ and take the corresponding branch of the tree. The process
repeats until a leaf is reached, at which point both parties learn
$f(x|_V\oplus w).$
\qquad
\end{proof}


\section{Pattern matrix method using threshold weight} 
\label{sec:additional-results}

As we have already mentioned,
Theorem~\ref{thm:main-cc} of the previous
section can be used to obtain lower bounds not only for bounded-error
communication but also small-bias communication. In the latter case,
one first needs to show that the base function $f\colon\zoo^t\to\moo$
cannot be approximated pointwise within $1-o(1)$ by a real polynomial
of a given degree $d.$ In this section, we derive a different lower
bound for small-bias communication, this time using the assumption
that the threshold weight $W(f,d)$ is high. We will see that this
new lower bound is nearly optimal and closely related to the lower
bound in Theorem~\ref{thm:main-cc}.

\newpage

\begin{theorem}[pattern matrix method using threshold weight]
Let $F$ be the $(n,t,f)$-pattern matrix, 
where $f\colon \zoo^t\to\moo$ is given.
Then for every integer $d\geq1$ and real $\gamma\in(0,1),$
\begin{align}
Q^*_{1/2-\gamma/2}(F) &\geq
\frac14\min\left\{
 d\log \frac nt, \;\, \log \frac {W(f,d-1)}{2t}
   \right\} - \frac12 \log \frac 3\gamma.
\label{eqn:quantum-weight}
\intertext{In particular,}
Q^*_{1/2-\gamma/2}(F) &\geq
\frac 14\degthr(f)\log \left(\frac nt\right) - 
   \frac12 \log \frac 3\gamma.
\label{eqn:quantum-threshold-deg}
\end{align}
\label{thm:pattern-matrix-method-small-bias}
\end{theorem}

\begin{proof}
Letting $d=\degthr(f)$ in (\ref{eqn:quantum-weight}) yields
(\ref{eqn:quantum-threshold-deg}), since $W(f,d-1)=\infty$ in that
case.  In the remainder of the proof, we focus on
(\ref{eqn:quantum-weight}) alone.

We claim that there exists a distribution $\mu$ on $\zoo^t$ such that
\begin{align}
\max_{|S|<d} \left\lvert \Exp_{z\sim\mu}[f(z)\chi_S(z)]\right\rvert
	&\leq \left(\frac{2t}{W(f,d-1)}\right)^{1/2}.
	\label{eqn:ortho-distr}
\end{align}
For $d\leq\degthr(f),$ the claim holds by
Theorem~\ref{thm:thresh-or-distr} since $W(f,d-1)=\infty$ in that case.
For $d>\degthr(f),$ the claim holds by Theorem~\ref{thm:dual-weight}.

Now, define $\psi\colon\zoo^t\to\Re$ by $\psi(z)=f(z)\mu(z).$ 
It follows from (\ref{eqn:ortho-distr}) that
\begin{align}
&\;\,|\hat\psi(S)|\leq 2^{-t}\left(\frac{2t}{W(f,d-1)}\right)^{1/2}
 &&(|S|<d),
	\label{eqn:weight-psi-fourier-coeffs} \\
&\sum_{z\in\zoo^t}|\psi(z)| =1,
	\label{eqn:weight-psi-bounded} \\
&\sum_{z\in\zoo^t}\psi(z)f(z) =1.
	\label{eqn:weight-psi-correl}
\end{align}
Let $\Psi$ be the
$(n,t,2^{-n}(n/t)^{-t}\psi)$-pattern matrix. Then 
(\ref{eqn:weight-psi-bounded}) and (\ref{eqn:weight-psi-correl})
show that
\begin{equation}
\|\Psi\|_1 =1,  \qquad    \langle F, \Psi\rangle = 1.
\label{eqn:weight-K-bounded,K-M-correl}
\end{equation}
It remains to calculate $\|\Psi\|.$ By (\ref{eqn:weight-psi-bounded})
and Proposition~\ref{prop:fourier-coeff-bound},
\begin{equation}
\max_{S\subseteq[t]} |\hat\psi(S)| \leq 2^{-t}.
	\label{eqn:weight-max-fourier-coeff-psi}
\end{equation}
Theorem~\ref{thm:pattern-spectrum} yields, in view of
(\ref{eqn:weight-psi-fourier-coeffs}) and
(\ref{eqn:weight-max-fourier-coeff-psi}):
\begin{equation}
\|\Psi\|
\leq  \max\left\{\left(\frac{t}{n}\right)^{d/2},
\left(\frac{2t}{W(f,d-1)}\right)^{1/2}
         \right\}
      \left(2^{n+t} \left(\frac{n}{t}\right)^t\right)^{-1/2}.
\label{eqn:weight-K-norm}
\end{equation}
Now (\ref{eqn:quantum-weight}) follows from 
(\ref{eqn:weight-K-bounded,K-M-correl}),
(\ref{eqn:weight-K-norm}), and
Theorem~\ref{thm:discrepancy-method}.
\qquad
\end{proof}

Recall from Theorem~\ref{thm:E-vs-W} that the quantities $E(f,d)$
and $W(f,d)$ are related for all $f$ and $d.$ In particular,
the lower bounds for small-bias communication in
Theorems~\ref{thm:main-cc}
and~\ref{thm:pattern-matrix-method-small-bias} are quite close, and
either one can be approximately deduced from the other. In deriving
both results from scratch, as we did, our motivation was to obtain
the tightest bounds and to illustrate the pattern matrix method in
different contexts. We will now see that the lower bound in
Theorem~\ref{thm:pattern-matrix-method-small-bias} is close to
optimal, even for classical protocols.

\begin{theorem}
Let $F$ be the $(n,t,f)$-pattern matrix, 
where $f\colon \zoo^t\to\moo$ is given.
Then for every integer $d\geq\degthr(f),$ 
\begin{align*}
Q^*_{1/2-\gamma/2}(F) &\leq R_{1/2-\gamma/2}(F)
\leq d\log\left(\frac nt\right) + 3,
\end{align*}
where $\gamma = 1/W(f,d).$
\label{thm:classical-upper-bound}
\end{theorem}

\begin{proof}
The communication protocol that we will describe is standard and
has been used in one form or another in several works,
e.g.,~\cite{paturi86cc, GHR92, sherstov07halfspace-mat,
sherstov07ac-majmaj}.  Put $W=W(f,d)$ and fix a representation 
\begin{align*}
f(z) \equiv
\sign\left(
	\sum_{S\subseteq[t], \; |S|\leq d} \lambda_S \chi_S(z)
\right),
\end{align*}
where the integers $\lambda_S$ satisfy $\sum |\lambda_S|=W.$ On
input $(x,(V,w)),$ the protocol proceeds as follows. Let
$i_1<i_2<\cdots<i_t$ be the elements of $V.$ Alice and Bob use their
shared randomness to pick a set $S\subseteq[t]$ with $|S|\leq d,$
according to the probability distribution $|\lambda_S|/W.$ Next,
Bob sends Alice the indices $\{i_j:j\in S\}$ as well as the bit
$\chi_S(w).$ With this information, Alice computes the product
$\sign(\lambda_S)\chi_S(x|_V)\chi_S(w)=\sign(\lambda_S)\chi_S(x|_V\oplus
w)$ and announces the result as the output of the protocol.

Assuming an optimal encoding of the messages, the communication cost
of this protocol is bounded by 
\begin{align*}
\left\lceil
\log \left(\frac nt\right)^d 
\right\rceil +2
\leq d\log\left(\frac nt\right) + 3,
\end{align*}
as desired.  On each input $x,V,w,$ the output of the protocol is a
random variable $P(x,V,w)\in\moo$ that obeys
\begin{align*}
f(x|_V\oplus w)\Exp[P(x,V,w)] 
    &= f(x|_V\oplus w)\sum_{|S|\leq d} \frac{|\lambda_S|}{W} 
		\sign(\lambda_S)\chi_S(x|_V\oplus w)\\
    &= \frac1W
	  \left|\sum_{|S|\leq d} \lambda_S\chi_S(x|_V\oplus w) \right|\\
    &\geq \frac 1W,
\end{align*}
which means that the protocol produces the correct answer with probability
$\frac 12  + \frac1{2W}$ or greater.
\qquad
\end{proof}

\section{Discrepancy of pattern matrices} \label{sec:disc}
We now restate some of the results of the previous section
in terms of \emph{discrepancy}, a key notion already mentioned
in Section~\ref{sec:discrepancy}. This quantity figures
prominently in the study of small-bias communication as well as
various applications, such as  learning theory and circuit complexity.

For a Boolean function
$f\colon X\times Y\to\moo$ and a probability distribution $\lambda$ on
$X\times Y,$ the \emph{discrepancy} of $f$ under
$\lambda$ is defined by
\begin{align*}
 \disc_\lambda(f) &=
	   \max_{\rule{0mm}{4mm}\substack{S\subseteq X,\\T\subseteq Y}}
    \left|\sum_{x\in S} \sum_{y\in T} \lambda(x,y) f(x,y)\right|.
\intertext{We put}
 \disc(f) &= \min_{\lambda} \disc_\lambda(f). 
\end{align*} 
As usual, we will identify a 
function $f\colon X\times Y\to\moo$ with its communication
matrix~$F=[f(x,y)]_{x,y}$ and use the conventions
$\disc_\lambda(F)=\disc_\lambda(f)$ and $\disc(F)=\disc(f).$

The above definition of discrepancy is not convenient to work with, and
we will use a well-known matrix-analytic reformulation; cf.~Kushilevitz
\& Nisan~\cite[Ex.~3.29]{ccbook}. For matrices $A=[A_{xy}]$ and
$B=[B_{xy}],$ recall that their \emph{Hadamard product} is given
by $A\circ B=[A_{xy}B_{xy}].$

\begin{proposition}
Let $X,Y$ be finite sets, $f\colon X\times Y\to\moo$ 
a given function. Then
\[ \disc_P(f) \leq \sqrt{|X|\, |Y|}\, \| P\circ F\|, \]
where $F=[f(x,y)]_{x\in X,\, y\in Y}$ and $P$ is any matrix whose
entries are nonnegative and sum to~$1$ $($viewed as a probability
distribution$).$ In particular,
\[ \disc(f) \leq \sqrt{|X|\, |Y|} \min_{P} \| P\circ F\|, \]
where the minimum is over matrices $P$ whose entries are
nonnegative and sum to $1.$
\label{prop:disc2spectral}
\end{proposition}

\begin{proof}
We have
\begin{align*}
\disc_P(f) 
	&= \max_{S,T} \left|\1_S\tr\, (P\circ F)\, \1_T \right| \\
    &\leq \max_{S,T}
		 \Big\{ \|\1_S\|\cdot \|P\circ F\|\cdot \|\1_T\| \Big\}\\
	&= \|P\circ F\| \sqrt{|X|\, |Y|},
\end{align*}
as claimed.
\qquad
\end{proof}

We will need one last ingredient, a well-known 
lower bound on communication complexity in terms of discrepancy.

\begin{proposition}[see {\cite[pp. 36--38]{ccbook}}] 
For every function $f\colon X\times Y\to\moo$ and
every $\gamma\in(0,1),$
\begin{align*}
R_{1/2-\gamma/2}(f) 
\geq 
\log\frac{\gamma}{\disc(f)}.
\end{align*}
\label{prop:disc-method}
\end{proposition}

Using Theorems~\ref{thm:pattern-matrix-method-small-bias} and
\ref{thm:classical-upper-bound}, we will now characterize the
discrepancy of pattern matrices in terms of threshold weight.

\begin{theorem}[discrepancy of pattern matrices]
Let $F$ be the $(n,t,f)$-pattern matrix, 
where $f\colon \zoo^t\to\moo$ is given.
Then for every integer $d\geq0,$
\begin{align}
\disc(F) &\geq \frac1{8W(f,d)} \left( \frac tn \right)^d
      \label{eqn:disc-lower}
\intertext{and}
\disc(F)^2 &\leq \max\left\{ \frac {2t}{W(f,d-1)}, 
                             \left(\frac tn\right)^d
             \right\}.
       \label{eqn:disc-upper}
\intertext{In particular,}
\disc(F) &\leq \left(\frac tn\right)^{\degthr(f)/2}.
       \label{eqn:disc-upper-thrdeg}
\end{align}
\label{thm:pattern-matrix-discrepancy}
\end{theorem}

\begin{proof}
The lower bound (\ref{eqn:disc-lower}) is immediate from
Theorem~\ref{thm:classical-upper-bound} and
Proposition~\ref{prop:disc-method}. For the upper
bound~(\ref{eqn:disc-upper}), construct the matrix $\Psi$ as in the
proof of Theorem~\ref{thm:pattern-matrix-method-small-bias}.
Then~(\ref{eqn:weight-K-bounded,K-M-correl}) shows that $\Psi= F\circ
P$ for a nonnegative matrix $P$ whose entries sum to~$1.$ As a
result,~(\ref{eqn:disc-upper}) follows from~(\ref{eqn:weight-K-norm})
and Proposition~\ref{prop:disc2spectral}.
Finally, (\ref{eqn:disc-upper-thrdeg}) follows by taking $d=\degthr(f)$ in 
(\ref{eqn:disc-upper}), since $W(f,d-1)=\infty$ in that case.
\qquad
\end{proof}

This settles Theorem~\ref{thm:main-discrepancy} from the introduction.
Theorem~\ref{thm:pattern-matrix-discrepancy} follows up and considerably
improves on our earlier result, the \emph{Degree/Discrepancy
Theorem}~\cite{sherstov07ac-majmaj}:

\begin{theorem}[Sherstov]\label{thm:deg2disc}
Let $f\colon\zoo^t\to\moo$ be given. Fix an integer $n\geq t.$ 
Let $M=[f(x|_S)]_{x,S},$ where the row
index $x$ ranges over $\zoon$ and the column index $S$ ranges over all
$t$-element subsets of $\{1,2,\dots,n\}.$ Then
\begin{align*}
 \disc(M) \leq \left(\frac{4\e t^2}{n\degthr(f)}\right)^{\degthr(f)/2}.
\end{align*}
\end{theorem}

Note that (\ref{eqn:disc-upper-thrdeg}) is already stronger than 
Theorem~\ref{thm:deg2disc}. In Section~\ref{sec:app-discrepancy},
we will see an example when Theorem~\ref{thm:pattern-matrix-discrepancy}
gives an exponential improvement on Theorem~\ref{thm:deg2disc}.

Threshold weight is typically easier to analyze than the
approximate degree.  For completeness, however, we will now
supplement Theorem~\ref{thm:pattern-matrix-discrepancy} with an
alternate bound on the discrepancy of a pattern matrix in terms
of the approximate degree. 

\begin{theorem}
\label{thm:pattern-matrix-discrepancy-adeg}
Let $F$ be the $(n,t,f)$-pattern matrix, 
for a given function $f\colon \zoo^t\to\moo.$
Then for every $\gamma>0,$
\begin{align*}
\disc(F) \leq \gamma + \left(\frac tn\right)^{\deg_{1-\gamma}(f)/2}.
\end{align*}
\end{theorem}

\begin{proof}
Let $d=\deg_{1-\gamma}(f)\geq1.$ Define $\epsilon=1-\gamma$ and
construct the matrix $\Psi$ as in the proof of Theorem~\ref{thm:main-cc}.
Then~(\ref{eqn:K-bounded,K-M-correl}) shows that $\Psi= H\circ P,$
where $H$ is a sign matrix and $P$ is a nonnegative matrix whose
entries sum to~$1.$ Viewing $P$ as a probability distribution, we
infer from~(\ref{eqn:K-norm}) and Proposition~\ref{prop:disc2spectral}
that
\begin{align}
	\disc_P(H) \leq \left(\frac tn\right)^{d/2}.
	\label{eqn:disc-P-H-adeg}
\end{align}
Moreover,
\begin{align}
  \disc_P(F) &\leq \disc_P(H) + \|(F-H)\circ P\|_1  \nonumber\\
  &=\disc_P(H) + 1 - \langle F,H\circ P\rangle  \nonumber\\
  &\leq \disc_P(H) + \gamma,
	\label{eqn:disc-P-F-adeg}
\end{align}
where the last step follows because $\langle F,\Psi\rangle
>\epsilon=1-\gamma$ by (\ref{eqn:K-bounded,K-M-correl}).
The proof is complete in view of (\ref{eqn:disc-P-H-adeg}) and
(\ref{eqn:disc-P-F-adeg}).
\qquad
\end{proof}

\section{Approximate rank and trace norm of pattern matrices}
\label{sec:approx-rank}
We will now use the results of the previous sections to analyze the
approximate rank and approximate trace norm of pattern matrices.
These notions were originally motivated by lower bounds on quantum
communication~\cite{yao93quantum, buhrman-dewolf01polynomials,
razborov03quantum}. However, they 
also arise in learning theory~\cite{colt07rankeps} and
are natural matrix-analytic quantities in their own right.  In
particular, Klivans and Sherstov~\cite{colt07rankeps} proved
exponential lower bounds on the approximate rank of disjunctions,
majority functions, and decision lists, with applications to agnostic
learning. In what follows, we broadly generalize these results to
\emph{any} functions with high approximate degree or high threshold
weight.

\begin{theorem}
Let $F$ be the $(n,t,f)$-pattern matrix, 
where $f\colon \zoo^t\to\moo$ is given.
Let $s=2^{n+t}(n/t)^t$ be the number of entries in $F.$
Then for every $\epsilon\in[0,1)$ and every $\delta\in[0,\epsilon],$
\begin{align}
\|F\|_{\Sigma,\delta} \geq (\epsilon - \delta) \left(\frac
nt\right)^{\degeps(f)/2}\sqrt s
\label{eqn:bounded-error-trace-norm}
\end{align}
and
\begin{align}
\rk_\delta F \geq 
\left(\frac{\epsilon - \delta}{1 + \delta}\right)^2 
\left(\frac nt\right)^{\degeps(f)}.
\label{eqn:bounded-error-approx-rank}
\end{align}
\label{thm:bounded-error-approx-rank}
\end{theorem}

\begin{proof}
We may assume that $\degeps(f)\geq1,$ since otherwise $f$ is a
constant function and the claims hold trivially.
Construct $\Psi$ as in the proof of
Theorem~\ref{thm:main-cc}. Then the claimed
lower bound on $\|F\|_{\Sigma,\delta}$ follows from
(\ref{eqn:K-bounded,K-M-correl}), (\ref{eqn:K-norm}), and
Proposition~\ref{prop:approximate-trace-norm}.
Finally, (\ref{eqn:bounded-error-approx-rank}) follows from
(\ref{eqn:bounded-error-trace-norm}) and
Proposition~\ref{prop:approx-rank-approx-trace-norm}.
\qquad
\end{proof}

We prove an additional lower bound in the case of small-bias approximation.

\begin{theorem}
Let $F$ be the $(n,t,f)$-pattern matrix, 
where $f\colon \zoo^t\to\moo$ is given.
Let $s=2^{n+t}(n/t)^t$ be the number of entries in $F.$
Then for every $\gamma\in(0,1)$ and every integer $d\geq1,$
\begin{align}
\|F\|_{\Sigma,1-\gamma} &\geq \gamma 
		\min\left\{\left(\frac nt\right)^{d/2},
		\left(\frac{W(f,d-1)}{2t}\right)^{1/2}
				 \right\} \sqrt s
\label{eqn:small-bias-trace-norm}
\intertext{and}
\rk_{1-\gamma} F &\geq 
		\left(\frac{\gamma}{2-\gamma}\right)^2
		\min\left\{\left(\frac nt\right)^d,
		\frac{W(f,d-1)}{2t}
				 \right\}.
\label{eqn:small-bias-approx-rank}
\intertext{In particular,}
\|F\|_{\Sigma,1-\gamma} &\geq \gamma 
		 \left(\frac nt\right)^{\degthr(f)/2} \sqrt s
\label{eqn:degthr-trace-norm}
\intertext{and}
\rk_{1-\gamma} F &\geq 
		\left(\frac{\gamma}{2-\gamma}\right)^2
			\left(\frac nt\right)^{\degthr(f)}.
\label{eqn:degthr-approx-rank}
\end{align}
\label{thm:small-bias-approx-rank}
\end{theorem}

\begin{proof}
Construct $\Psi$ as in the proof of
Theorem~\ref{thm:pattern-matrix-method-small-bias}. Then the claimed
lower bound on $\|F\|_{\Sigma,\delta}$ follows from
(\ref{eqn:weight-K-bounded,K-M-correl}), (\ref{eqn:weight-K-norm}),
and Proposition~\ref{prop:approximate-trace-norm}.
Now (\ref{eqn:small-bias-approx-rank}) follows from
(\ref{eqn:small-bias-trace-norm}) and
Proposition~\ref{prop:approx-rank-approx-trace-norm}.
Finally, (\ref{eqn:degthr-trace-norm}) and (\ref{eqn:degthr-approx-rank})
follow by taking $d=\degthr(f)$ in (\ref{eqn:small-bias-trace-norm})
and (\ref{eqn:small-bias-approx-rank}), respectively, since
$W(f,d-1)=\infty$ in that case.
\qquad
\end{proof}

Theorems~\ref{thm:bounded-error-approx-rank} 
and~\ref{thm:small-bias-approx-rank} settle
Theorem~\ref{thm:main-approx-rank} from the introduction.

Recall that Theorem~\ref{thm:pattern-spectrum} gives an easy way to
calculate the trace norm and rank of a pattern matrix. 
In particular, it is straightforward to verify that the lower bounds 
in~(\ref{eqn:bounded-error-approx-rank}) 
and (\ref{eqn:small-bias-approx-rank}) 
are close to optimal for various
choices of $\epsilon,\delta,\gamma.$ For example, one has
$\|F-A\|_{\infty}\leq 1/3$ by taking $F$ and $A$ to be the $(n,t,f)$- and
$(n,t,\phi)$-pattern matrices, where $\phi\colon\zoo^t\to\Re$ is any
polynomial of degree $\deg_{1/3}(f)$ with $\|f-\phi\|_\infty\leq1/3.$

\section{Application: quantum complexity of symmetric functions}
\label{sec:razborovs-result}

As an illustrative application of the pattern matrix method, we now give a
short and elementary proof of Razborov's optimal lower bounds for every
predicate $D\colon \zodn\to\moo.$
We first solve the problem for all predicates $D$ that change value
close to~$0.$ Extension to the general case will require an additional
step.

\begin{theorem}
Let $D\colon \zodn\to\moo$ be a given predicate. 
Suppose that $D(\ell)\ne D(\ell-1)$ 
for some $\ell\leq \oneeighth n.$ 
Then \[ \qcc(D) \geq \Omega(\sqrt{n\ell}). \]
\label{thm:quantum-l-small}
\end{theorem}

\begin{proof}
It suffices to show that $Q^*_{1/7}(D)\geq\Omega(\sqrt{n\ell}).$
Define $f\colon \zoo^{\lfloor n/4\rfloor}\to\moo$ by $f(z)=D(|z|).$
Then $\adeg(f) \geq \Omega(\sqrt{n\ell})$ by Theorem~\ref{thm:paturi}.
Theorem~\ref{thm:main-cc} implies that
\[ Q^*_{1/7}(F) \geq \Omega(\sqrt{n\ell}), \]
where $F$ is the $(2\lfloor n/4\rfloor,\lfloor n/4\rfloor,f)$-pattern
matrix.  Since $F$ occurs as a submatrix of $[D(|x\wedge y|)]_{x,y},$  
the proof is complete.
\qquad
\end{proof}

The remainder of this section is a simple if tedious exercise in
shifting and padding. We note that Razborov's proof concludes in a
similar way (see~\cite{razborov03quantum}, beginning of Section~5).

\begin{theorem}
Let $D\colon \zodn\to\moo$ be a given predicate. 
Suppose that $D(\ell)\ne D(\ell-1)$ 
for some $\ell>\oneeighth n.$ Then 
\begin{equation}
\qcc(D) \geq c(n-\ell)
\label{eqn:quantum-l-large}
\end{equation}
for some absolute constant $c>0.$
\label{thm:quantum-l-large}
\end{theorem}

\begin{proof}
Consider the communication problem of computing $D(|x\wedge y|)$
when the last $k$ bits in $x$ and $y$ are fixed to $1.$ In other
words, the new problem is to compute $D_{k}(|x'\wedge y'|),$ where
$x',y'\in\zoo^{n-k}$ and the predicate $D_{k}\colon \{0,1,\dots,n-k\}\to\moo$
is given by $D_{k}(i) \equiv D(k+i).$
Since the new problem is a restricted version of the original, we have
\begin{equation}
\qcc(D) \geq \qcc(D_{k}).
\label{eqn:shift-cc}
\end{equation}
We complete the proof by placing a lower bound on $\qcc(D_{k})$ for 
\[ k   =    \ell - 
\left\lfloor \frac{\alpha}{1-\alpha}\cdot (n-\ell)\right\rfloor,\]
where $\alpha=\oneeighth.$ Note that 
$k$ is an integer between $1$ and $\ell$ (because $\ell>\alpha
n$).  The equality $k=\ell$ occurs if and only if $\big\lfloor
\frac{\alpha}{1-\alpha}(n-\ell)\big\rfloor=0,$ in which case
(\ref{eqn:quantum-l-large}) holds trivially for 
$c$ suitably small. Thus, we can
assume that $1\leq k\leq \ell-1,$ in which case $D_{k}(\ell-k)\ne
D_{k}(\ell-k-1)$ and $\ell-k \leq \alpha (n-k).$ Therefore,
Theorem~\ref{thm:quantum-l-small} is applicable to $D_{k}$ and
yields:
\begin{equation}
\qcc(D_{k}) \geq C\sqrt{(n-k)(\ell-k)},
\end{equation}
where $C>0$ is an absolute constant. Calculations reveal:
\begin{equation}
n-k = \left\lfloor \frac{1}{1-\alpha}\cdot (n-\ell) \right\rfloor,
\qquad\quad
\ell-k = \left\lfloor \frac{\alpha}{1-\alpha}\cdot (n-\ell) \right\rfloor.
\label{eqn:l-k-lower}
\end{equation}
The theorem is now immediate from 
(\ref{eqn:shift-cc})--(\ref{eqn:l-k-lower}).
\qquad
\end{proof}

Together, Theorems~\ref{thm:quantum-l-small} and~\ref{thm:quantum-l-large}
give the main result of this section:

\begin{restatetheorem}{thm:razborov03quantum}
Let $D\colon\zodn\to\moo.$ Then
\[ Q^*_{1/3}(D) \geq
\Omega(\sqrt{n\ell_0(D)}+\ell_1(D)), \]
where $\ell_0(D)\in\{0,1,\dots,\lfloor n/2\rfloor\}$ and
$\ell_1(D)\in\{0,1,\dots,\lceil n/2\rceil\}$
are the smallest integers such that $D$ is constant in the range
$[\ell_0(D),n-\ell_1(D)].$ 
\end{restatetheorem}

\begin{proof}
If $\ell_0(D)\ne0,$ set $\ell   =       \ell_0(D)$ and note that
$D(\ell)\ne D(\ell-1)$ by definition. One of Theorems~\ref{thm:quantum-l-small}
and~\ref{thm:quantum-l-large} must be applicable, and therefore
$\qcc(D)\geq \min\{ \Omega(\sqrt{n\ell}),\; \Omega(n-\ell)\}.$ 
Since $\ell \leq n/2,$ this simplifies to
\begin{equation}
\qcc(D) \geq \Omega(\sqrt{n\ell_0(D)}).  \label{eqn:l-1}
\end{equation}

If $\ell_1(D)\ne 0,$ set $\ell   =        n-\ell_1(D)+1\geq n/2$
and note that $D(\ell)\ne D(\ell-1)$ as before. By
Theorem~\ref{thm:quantum-l-large},
\begin{equation}
\qcc(D) \geq \Omega\left(\ell_1(D)\right).   \label{eqn:l-2}
\end{equation}
The theorem follows from (\ref{eqn:l-1}) and (\ref{eqn:l-2}).
\qquad
\end{proof}

\section{Application: discrepancy of constant-depth circuits}
\label{sec:app-discrepancy}

As another application of the pattern matrix method, we revisit the
discrepancy of $\AC^0,$ the class of polynomial-size constant-depth
circuits with AND, OR, NOT gates. In an earlier
work~\cite{sherstov07ac-majmaj}, we obtained the first exponentially
small upper bound on the discrepancy of a function in $\AC^0,$ with
applications to threshold circuits.  Independently, Buhrman et
al.~\cite{buhrman07pp-upp} exhibited another function in $\AC^0$
with exponentially small discrepancy.  We revisit these two discrepancy
bounds below, considerably sharpening the bound
in~\cite{sherstov07ac-majmaj} and giving a new and simple proof
of the bound in~\cite{buhrman07pp-upp}.

Consider the function $\mip_m\colon\zoo^{4m^3}\to\moo$ given by
\begin{align*}
\mip_m(x) = \bigvee_{i=1}^m\bigwedge_{j=1}^{4m^2}x_{ij}.
\end{align*}
This function was originally defined and studied by Minsky and
Papert~\cite{minsky88perceptrons} 
in their seminal monograph on perceptrons. Using this function
and the Degree/Discrepancy Theorem (Theorem~\ref{thm:deg2disc}),
an upper bound of $\exp\{-\Omega(n^{1/5})\}$ was derived
in~\cite{sherstov07ac-majmaj} on the discrepancy of an explicit
$\AC^0$ circuit $f\colon\zoon\times\zoon\to\moo$ of depth~$3.$
We will now sharpen that bound to
$\exp\{-\Omega(n^{1/3})\}.$

\begin{restatetheorem}{thm:main-mip}
Let $f(x,y)=\mip_m(x\vee y).$ Then 
\begin{align*}
\disc(f) = \exp\{-\Omega(m)\}. 
\end{align*}
\end{restatetheorem}

\begin{proof}
Put $d=\lfloor m/2\rfloor.$ A well-known result of Minsky and
Papert~\cite{minsky88perceptrons} states that $\degthr(\mip_d)\geq d.$
Since the $(8d^3,4d^3,\mip_d)$-pattern matrix is a submatrix of
$[f(x,y)]_{x,y},$ the proof is complete in view of 
equation (\ref{eqn:disc-upper-thrdeg}) of
Theorem~\ref{thm:pattern-matrix-discrepancy}.
\qquad
\end{proof}

We now turn to the result of Buhrman et al.
The ODD-MAX-BIT function $\OMB_n\colon\zoon\to\moo,$
due to Beigel~\cite{beigel94perceptrons}, is given by 
\begin{align}
\OMB_n(x) = \sign\left(1 + \sum_{i=1}^n (-2)^i x_i\right).
\label{eqn:omb-def}
\end{align}
It is straightforward to compute $\OMB_n$ by a linear-size DNF
formula and even a decision list.  In particular,
$\OMB_n$ belongs to the class $\AC^0.$ Buhrman et
al.~\cite[\S3.2]{buhrman07pp-upp} proved the following result.

\begin{theorem}[Buhrman et al.]
Let $f(x,y)=\OMB_n(x\wedge y).$ Then 
\begin{align*}
\disc(f) = \exp\{-\Omega(n^{1/3})\}.
\end{align*}
   \label{thm:BVW}
\end{theorem}

Using the results of this paper, we can give a short alternate proof of
this theorem.

\begin{proof}
Put $m=\lfloor n/4\rfloor.$ A well-known result due to
Beigel~\cite{beigel94perceptrons} shows that $W(\OMB_m,
cm^{1/3})\geq\exp(cm^{1/3})$ for some absolute constant $c>0.$
Since the $(2m,m,\OMB_m)$-pattern matrix is a submatrix of
$[f(x,y)]_{x,y},$ the proof is complete by
Theorem~\ref{thm:pattern-matrix-discrepancy}.
\qquad
\end{proof}

\begin{remark}{\rm
The above proofs illustrate that the characterization of the
discrepancy of pattern matrices in this paper
(Theorem~\ref{thm:pattern-matrix-discrepancy}) is a substantial
improvement on our earlier result (Theorem~\ref{thm:deg2disc}). In
particular, the representation (\ref{eqn:omb-def}) makes it clear
that $\degthr(\OMB_n)=1$ and therefore Theorem~\ref{thm:deg2disc}
cannot yield an upper bound better than $n^{-\Omega(1)}$ on the
discrepancy of $\OMB_n(x\wedge y).$
Theorem~\ref{thm:pattern-matrix-discrepancy}, on the other hand,
gives an exponentially better upper bound.}
\end{remark}

It is well-known~\cite{GHR92, hajnal93threshold-const-depth,
nisan93threshold} that the discrepancy of a function $f$ implies a
lower bound on the size of depth-$2$ majority circuits that compute
$f.$ Following~\cite{sherstov07ac-majmaj}, we record the consequences
of Theorems~\ref{thm:main-mip} and~\ref{thm:BVW} in this regard.

\begin{theorem}
Any majority vote of threshold gates that computes the function
\begin{align*}
f(x,y)=\mip_m(x\vee y)
\end{align*}
has size $\exp\{\Omega(m)\}.$ Analogously, any majority vote of
threshold gates that computes the function 
\begin{align*}
f(x,y)=\OMB_n(x\wedge y)
\end{align*}
has size $\exp\{\Omega(n^{1/3})\}.$
\end{theorem}

\begin{proof}
Analogous to the proof given in~\cite[\S7]{sherstov07ac-majmaj}.
\qquad
\end{proof}

\section{Pattern matrices and the log-rank conjecture} \label{sec:logrank}

In previous sections, we characterized various matrix-analytic
and combinatorial properties of pattern matrices, including their classical
and quantum communication complexity, discrepancy, approximate rank, and
approximate trace norm. We conclude this study with another interesting
fact about pattern matrices. Specifically, we show that they 
satisfy the well-known \emph{log-rank conjecture}~\cite[p.~26]{ccbook}.

In a seminal paper, Mehlhorn and Schmidt~\cite{mehlhorn-schmidt82rank-cc}
observed that the deterministic communication complexity of a sign
matrix $F$ satisfies $D(F)\geq \log \rk F.$ The log-rank
conjecture is that this lower bound is always tight up to a polynomial
factor, i.e., $D(F)\leq (\log\rk F)^{O(1)}.$ Using
the results of the previous sections, we can give a short proof of
this hypothesis in the case of pattern matrices.

\begin{theorem}[on the log-rank conjecture]
Let $f\colon\zoo^t\to\moo$ be a given function, $d=\deg(f).$ 
Let $F$ be the $(n,t,f)$-pattern matrix. Then 
\begin{align}
\rk F \geq \left(\frac{n}{t}\right)^{d}  \label{eqn:rk-F}
      \geq \exp\{\Omega(D(F)^{1/4})\}.   
\end{align}
In particular, $F$ satisfies the log-rank conjecture.
\end{theorem}

\begin{proof}
Since $\hat f(S)\ne 0$ for some set $S$ with $|S|=d,$ 
Theorem~\ref{thm:pattern-spectrum} implies that $F$ has at least
$(n/t)^{d}$ nonzero singular values. This settles the first inequality 
in~(\ref{eqn:rk-F}).

Proposition~\ref{prop:det-upper-bound} implies that $D(F)\leq
O(\dt(f)\log(n/t)),$ where $\dt(f)$ denotes the least depth of a decision tree
for $f.$ Nisan and Smolensky~\cite[Thm.~12]{buhrman-dewolf02DT-survey} 
prove that $\dt(f)\leq 2\deg(f)^4$ for all $f.$ Combining these two
observations establishes the second inequality in (\ref{eqn:rk-F}).
\qquad
\end{proof}

\section{Related work} \label{sec:shi-zhu}
Shi and Zhu~\cite{shi-zhu07block-composed-arxiv-updated} independently
obtained a result related to our lower
bound~(\ref{eqn:pattern-matrix-bounded-error}) on bounded-error
communication.  Fix functions $f\colon \zoon\to\moo$ and $g\colon
\zoo^k\times\zoo^k\to\zoo.$ Let $f\circ g^n$ denote the composition
of $f$ with $n$ independent copies of $g.$ More formally, the
function $f\circ g^n\colon \zoo^{nk}\times\zoo^{nk}\to\moo$ is given
by
	\[ (f\circ g^n)(x,y)    =    f\Big(\;g(x^{(1)},y^{(1)}),\;\; \dots,\;\; 
	g(x^{(n)},y^{(n)})\; \Big), \]
where $x=(x^{(1)},\dots,x^{(n)})\in\zoo^{nk}$ and
$y=(y^{(1)},\dots,y^{(n)})\in\zoo^{nk}.$ Shi and Zhu study the
communication complexity of $f\circ g^n.$ Their main
result~\cite[Lem.~3.5]{shi-zhu07block-composed-arxiv-updated} is 
that
\[ Q^*_{1/3} (f\circ g^n) \geq \Omega(\deg_{1/3}(f))
\qquad\text{provided that} \qquad
\rho(g)\leq \frac{\deg_{1/3}(f)}{2\e n}, 
\]
where $\rho(g)$ is a new variant of discrepancy that the authors introduce. 
As an illustration, they re-prove a weaker version of Razborov's lower bounds in
Theorem~\ref{thm:razborov03quantum}. In our terminology
(Section~\ref{sec:discrepancy}), their proof also fits in the
framework of the Klauck-Razborov generalized discrepancy method.

Shi and Zhu's result revolves around the quantity $\rho(g),$ which
needs to be small.  This poses two complications.
First, the function $g$ will generally need to depend on 
many variables, from $k=\Theta(\log n)$ to $k=n^{\Theta(1)},$
which weakens the final lower bounds on
communication.  For example, the lower bounds obtained
in~\cite{shi-zhu07block-composed-arxiv-updated} for symmetric functions are
polynomially weaker than optimal (Theorem~\ref{thm:razborov03quantum}).

A second complication, as the authors note, is that ``estimating
$\rho(g)$ is unfortunately difficult in
general"~\cite[\S4.1]{shi-zhu07block-composed-arxiv-updated}.
For example, re-proving Razborov's lower bounds
reduces to estimating $\rho(g)$ for $g(x,y)=x_1y_1\vee
\cdots\vee x_ky_k.$ Shi and Zhu accomplish this using 
Hahn matrices, an advanced tool that is the centerpiece
of Razborov's own proof (Razborov's use of Hahn matrices
is somewhat more demanding).

Our method avoids these complications altogether.  
For example, we prove (by taking $n=2t$ in
the pattern matrix method,
Theorem~\ref{thm:main-cc}) that
\[Q^*_{1/3} (f\circ g^n) \geq \Omega(\deg_{1/3}(f))\] 
for any function $g\colon \zoo^k\times\zoo^k\to\zoo$ such that the
matrix $[g(x,y)]_{x,y}$ contains the following submatrix, up to
permutations of rows and columns:
\[
\begin{bmatrix}
1 & 0 & 1 & 0 \\
1 & 0 & 0 & 1 \\
0 & 1 & 1 & 0 \\
0 & 1 & 0 & 1 
\end{bmatrix}.
\]
To illustrate, one can take $g$ to be
\[ g(x,y) \;\;=\;\; x_1y_1\;\;\vee\;\; x_2y_2\;\;\vee\;\;
x_3y_3\;\;\vee\;\; x_4y_4 \]
or 
\[ g(x,y)\;\;=\;\;
	x_1y_1y_2
		\;\;\vee\;\;
	\overline{x_1}\,y_1\overline{y_2}
		\;\;\vee\;\;
	x_2\,\overline{y_1}\,y_2
		\;\;\vee\;\;
	\overline{x_2}\,\overline{y_1}\,\overline{y_2}.
\]
In summary, there is a simple function $g$ on $k=2$ variables
that works universally for all $f.$ 
This means no technical conditions to check, such as
$\rho(g),$ and no blow-up in the number of variables.
As a result, we are able
to re-prove Razborov's optimal lower bounds exactly.  Moreover, the technical
machinery of this paper is self-contained and disjoint from Razborov's
proof.

A further advantage of the pattern matrix method is that it extends
in a straightforward way to the multiparty
model~\cite{lee-shraibman08disjointness, chatt-ada08disjointness, 
pitassi08np-rp, david-pitassi-viola08bpp-np,
beame-huyn-ngoc08multiparty-eccc}.  This extension depends on 
the fact that the rows of a pattern matrix are applications of
the same function to different subsets of the variables. In the
general context of block composition, it is unclear how to carry
out this extension.  Further details can be found in the
survey~\cite{dual-survey}.

These considerations do not diminish the technical merit of Shi and
Zhu's method, which is of much interest.  The proofs
in~\cite{shi-zhu07block-composed-arxiv-updated} and this paper start out with the
same duality transformation (Theorem~\ref{thm:dual-approx}) 
but diverge substantially from then on, which explains
the differences in our results.  Specifically, we introduce and
analyze pattern matrices, while Shi and Zhu construct a much different
family of matrices.  

\section*{Acknowledgments}

I would like to thank Adam Klivans, James Lee, Sasha Razborov,
Yaoyun Shi, Avi Wigderson, and Ronald de Wolf for their feedback
on an earlier version of this manuscript.  
I am thankful to Ronald de Wolf for his permission to include his
remark on Theorem~\ref{thm:main-cc}.  It is also to Ronald that I
owe my interest in quantum communication.  This research was supported
by Adam Klivans' NSF CAREER Award and NSF Grant CCF-0728536.

{
\bibliographystyle{siam2}
\bibliography{refs}
}

\appendix

\section{On uniform approximation and sign-representation}
\label{sec:E-vs-W}
The purpose of this appendix is to prove Theorem~\ref{thm:E-vs-W}
on the representation of a Boolean function by real versus integer
polynomials.  Similar statements have been noted earlier by several
authors~\cite{KP98threshold, buhrman07pp-upp}.  We derive our result
by modifying a recent analysis due to Buhrman et
al.~\cite[Cor.~1]{buhrman07pp-upp}.

\begin{restatetheorem}{thm:E-vs-W}
Let $f\colon \zoon\to\moo$ be given. Then for $d=0,1,\dots,n,$
\begin{align*}
\frac1{1-E(f,d)} \leq W(f,d) \leq \frac2{1-E(f,d)}
\left\{
{n\choose 0}+{n\choose 1}+\cdots+{n\choose d}\right\}^{3/2},
\end{align*}
with the convention that $1/0=\infty.$
\end{restatetheorem}

\begin{proof}
One readily verifies that $W(f,d)=\infty$ if and only if $E(f,d)=1.$
In what follows, we focus on the complementary case when $W(f,d)<\infty$
and $E(f,d)<1.$

For the lower bound on $W(f,d),$ fix integers $\lambda_S$
with $\sum_{|S|\leq d}|\lambda_S|=W(f,d)$ such that the polynomial
$p(x)=\sum_{|S|\leq d}\lambda_S\chi_S(x)$ satisfies $f(x)\equiv
\sign p(x).$ Then $1\leq f(x)p(x)\leq W(f,d)$ and therefore
\begin{align*}
E(f,d)\leq \left\|f-\frac1{W(f,d)}\,p\right\|_\infty\leq 1-\frac1{W(f,d)}.
\end{align*}

To prove the upper bound on $W(f,d),$ fix any degree-$d$ polynomial
$p$ such that $\|f-p\|_\infty=E(f,d).$ Define $\delta=1-E(f,d)>0$
and $N=\sum_{i=0}^d {n\choose i}.$ For a real $t,$ let $\rnd t$
be the result of rounding $t$ to the closest integer, so that $|t-\rnd
t|\leq 1/2.$ We claim that the polynomial
\begin{align*}
 q(x)=  \sum_{|S|\leq d} \rnd(M\hat p(S))\chi_S(x),
\end{align*}
where $M=3N/(4\delta),$ satisfies $f(x)\equiv \sign q(x).$
Indeed,
\begin{align*} 
\left|f(x) - \frac1M q(x)\right|
 &\leq |f(x) - p(x)| + \frac1M |Mp(x) - q(x)|\\
 &\leq 1 - \delta  + \frac1M \sum_{|S|\leq d} 
    |M\hat p(S) - \rnd(M\hat p(S))|\\
 &\leq 1 - \delta + \frac{N}{2M}\\
 &<1.
\end{align*}
It remains to examine the sum of the coefficients of $q.$ We have:
\begin{align*}
\sum_{|S|\leq d} |\rnd(M\hat p(S))|
&\leq \frac12 N  + M\sum_{|S|\leq d} |\hat p(S)|\\
&\leq \frac12 N  + M\left(N \Exp_x\left[p(x)^2\right]\right)^{1/2}\\
&\leq \frac{2N\sqrt N}\delta,
\end{align*}
where the second step follows by an application of 
the Cauchy-Schwarz inequality and Parseval's identity (\ref{eqn:parsevals}). 
\quad
\end{proof}

\end{document}